\documentclass[english]{article}
\usepackage[T1]{fontenc}
\usepackage[latin9]{inputenc}
\usepackage[a4paper]{geometry}
\usepackage{color}
\usepackage{babel}
\usepackage{float}
\usepackage{amsmath}
\usepackage{amsthm}
\usepackage{amssymb}
\usepackage{stmaryrd}
\usepackage{graphicx}
\usepackage{setspace}
\usepackage[authoryear]{natbib}
\usepackage[unicode=true,pdfusetitle,
 bookmarks=true,bookmarksnumbered=false,bookmarksopen=false,
 breaklinks=false,pdfborder={0 0 0},pdfborderstyle={},backref=false,colorlinks=true]
 {hyperref}
\hypersetup{
 linkcolor=blue,urlcolor=blue,citecolor=blue}

\makeatletter

\floatstyle{ruled}
\newfloat{algorithm}{tbp}{loa}
\providecommand{\algorithmname}{Algorithm}
\floatname{algorithm}{\protect\algorithmname}

 \theoremstyle{definition}
  \newtheorem{example}{\protect\examplename}
  \theoremstyle{remark}
  \newtheorem{rem}{\protect\remarkname}
  \theoremstyle{plain}
  \newtheorem{prop}{\protect\propositionname}
\theoremstyle{plain}
\newtheorem{thm}{\protect\theoremname}
  \theoremstyle{plain}
  \newtheorem{cor}{\protect\corollaryname}
  \theoremstyle{plain}
  \newtheorem{lem}{\protect\lemmaname}

\@ifundefined{showcaptionsetup}{}{%
 \PassOptionsToPackage{caption=false}{subfig}}
\usepackage{subfig}
\makeatother

  \providecommand{\examplename}{Example}
  \providecommand{\lemmaname}{Lemma}
  \providecommand{\propositionname}{Proposition}
  \providecommand{\remarkname}{Remark}
\providecommand{\corollaryname}{Corollary}
\providecommand{\theoremname}{Theorem}

\begin{document}

\title{Unbiased approximations of products of expectations}

\author{Anthony Lee\thanks{School of Mathematics, University of Bristol and The Alan Turing Institute}~,
Simone Tiberi\thanks{Institute of Molecular Life Sciences and SIB Swiss Institute of Bioinformatics,
University of Zurich}~, and Giacomo Zanella\thanks{Department of Decision Sciences, BIDSA and IGIER, Bocconi University}}
\maketitle
\begin{abstract}
We consider the problem of approximating the product of $n$ expectations
with respect to a common probability distribution $\mu$. Such products
routinely arise in statistics as values of the likelihood in latent
variable models. Motivated by pseudo-marginal Markov chain Monte Carlo
schemes, we focus on unbiased estimators of such products. The standard
approach is to sample $N$ particles from $\mu$ and assign each particle
to one of the expectations. This is wasteful and typically requires
the number of particles to grow quadratically with the number of expectations.
We propose an alternative estimator that approximates each expectation
using most of the particles while preserving unbiasedness. We carefully
study its properties, showing that in latent variable contexts the
proposed estimator needs only $\mathcal{O}(n)$ particles to match
the performance of the standard approach with $\mathcal{O}(n^{2})$
particles. We demonstrate the procedure on two latent variable examples
from approximate Bayesian computation and single-cell gene expression
analysis, observing computational gains of the order of the number
of expectations, i.e.\ data points, $n$.
\end{abstract}

\section{Introduction}

Let $X$ be a random variable with probability measure, or distribution,
$\mu$ on a measurable space $(E,\mathcal{E})$, and $L^{1}(\mu)$
the class of integrable, real-valued functions, i.e. $L^{1}(\mu)=\{f:\int_{E}\left|f(x)\right|\mu({\rm d}x)<\infty\}$.
For a sequence of non-negative ``potential'' functions $G_{1},\ldots G_{n}\in L^{1}(\mu)$,
we consider approximations of products of $n$ expectations
\begin{equation}
\gamma:=\prod_{p=1}^{n}\mathbb{E}\left[G_{p}(X)\right]=\prod_{p=1}^{n}\mu(G_{p}),\label{eq:gamma}
\end{equation}
where we denote $\mu(f):=\int_{E}f(x)\mu({\rm d}x)$ for $f\in L^{1}(\mu)$.
These arise, e.g., as values of the likelihood function in latent
variable models. We concentrate on unbiased approximations of $\gamma$;
these can be used, e.g., within pseudo-marginal Markov chain methods
for approximating posterior expectations \citep{Beaumont2003,Andrieu2009}.
To motivate this general problem, and because our main result in the
sequel relates to latent variable models, we provide the following
generic example of such a model.
\begin{example}
\label{example:latent-var-model}Let $g$ be a Markov transition density
and $Y_{1},\ldots,Y_{n}$ be i.i.d. $\mathsf{Y}$-valued random variables
distributed according to the probability density function $\nu$ where
\[
\nu(y):=\mathbb{E}\left[g(X,y)\right]=\int_{E}g(x,y)\mu({\rm d}x),\qquad y\in\mathsf{Y}.
\]
That is, the $Y_{p}$ are independent and distributed according to
$g(X_{p},\cdot)$ where $X_{p}\sim\mu$. For observations $y_{1},\ldots,y_{n}$,
respectively, of $Y_{1},\ldots,Y_{n}$, we can write
\[
\prod_{p=1}^{n}\nu(y_{p})=\prod_{p=1}^{n}\mathbb{E}\left[g(X,y_{p})\right]=\prod_{p=1}^{n}\mathbb{E}\left[G_{p}(X)\right]=\gamma,
\]
where the potential functions are defined via $G_{p}(x):=g(x,y_{p})$,
for $p\in\{1,\ldots,n\}$.
\end{example}
\begin{rem}
\label{rem:likelihood-function-lvarmod}To see that $\gamma$ can
be viewed as a value of the likelihood function, let $\theta\in\Theta$
be a statistical parameter and $\{(\mu_{\theta},g_{\theta}):\theta\in\Theta\}$
parameterized families of distributions and Markov transition densities.
The likelihood function $L$ is then $L(\theta):=\prod_{p=1}^{n}\nu_{\theta}(y_{p})$
where 
\[
\nu_{\theta}(y):=\mathbb{E}\left[g_{\theta}(X,y)\right]=\int_{E}g_{\theta}(x,y)\mu_{\theta}({\rm d}x),\qquad y\in\mathsf{Y},
\]
and clearly $L(\theta)$ is of the form (\ref{eq:gamma}) for any
$\theta\in\Theta$.
\end{rem}
The focus of this paper is approximations of $\gamma$ using $N$
independent and $\mu$-distributed random variables $\zeta:=(\zeta_{1},\ldots,\zeta_{N})$,
which we will refer to throughout as particles. A straightforward
approach to constructing an unbiased approximation of $\gamma$ is
to approximate each expectation $\mathbb{E}\left[G_{p}(X)\right]=\mu(G_{p})$
independently using $M$ particles, where $N=Mn$. To be precise,
we define 
\begin{equation}
\gamma_{{\rm simple}}^{N}:=\prod_{p=1}^{n}\frac{1}{M}\sum_{i=1}^{M}G_{p}(\zeta_{(p-1)M+i}).\label{eq:simple_approx}
\end{equation}
We will often refer to the second moment condition
\begin{equation}
\max_{p\in\{1,\ldots,n\}}\mu(G_{p}^{2})<\infty,\label{eq:second_moment_condition}
\end{equation}
and in order to simplify the presentation we define the normalized
sequence of potential functions $\bar{G}_{1},\ldots,\bar{G}_{n}$
via $\bar{G}_{p}:=G_{p}/\mu(G_{p})$. The following lack-of-bias,
consistency, second moment and variance properties are easily established.
We denote convergence in probability by $\overset{P}{\rightarrow}$.
\begin{prop}
\label{prop:simple_analysis}We have $\mathbb{E}\left[\gamma_{{\rm simple}}^{N}\right]=\gamma$,
$\gamma_{{\rm simple}}^{N}\overset{P}{\rightarrow}\gamma$ as $N\rightarrow\infty$
and 
\begin{equation}
\mathbb{E}\left[\left(\gamma_{{\rm simple}}^{N}/\gamma\right)^{2}\right]=\prod_{p=1}^{n}\left\{ 1+\left[\mu(\bar{G}_{p}^{2})-1\right]/M\right\} ,\label{eq:relvar_simple_approx}
\end{equation}
so ${\rm var}(\gamma_{{\rm simple}}^{N})$ is finite and converges
to $0$ as $M\rightarrow\infty$ if and only if (\ref{eq:second_moment_condition})
holds.
\end{prop}
The approximation $\gamma_{{\rm simple}}^{N}$ is straightforward
to compute and analyze since it is a product of averages of independent
random variables. However, each particle is only used to approximate
one of the expectations in the product, and in situations where these
particles are expensive to obtain this may be wasteful. An alternative
approach is to use
\begin{equation}
\gamma_{{\rm biased}}^{N}:=\prod_{p=1}^{n}\frac{1}{N}\sum_{i=1}^{N}G_{p}(\zeta_{i}),\label{eq:biased_approx}
\end{equation}
which is consistent and not wasteful, but also not unbiased in general.
\begin{prop}
\label{prop:biased_analysis}We have $\gamma_{{\rm biased}}^{N}\overset{P}{\rightarrow}\gamma$
as $N\rightarrow\infty$ but $\mathbb{E}\left[\gamma_{{\rm biased}}^{N}\right]\neq\gamma$
in general.
\end{prop}
We propose in the sequel an approximation $\gamma_{{\rm recycle}}^{N}$
that is unbiased like $\gamma_{{\rm simple}}^{N}$ but which is closer
to $\gamma_{{\rm biased}}^{N}$ in that it uses most of the particles
to approximate each expectation in the product while remaining computationally
tractable. The approximation $\gamma_{{\rm recycle}}^{N}$ can be
viewed as an unbiased approximation of $\gamma_{{\rm perm}}^{N}$,
the rescaled permanent of a particular rectangular matrix of random
variables, which is never worse in terms of variance than $\gamma_{{\rm simple}}^{N}$
but is very computationally costly to compute in general. The approximation
$\gamma_{{\rm recycle}}^{N}$ is an extension of the importance sampling
approximation of the permanent of a square matrix proposed by \citet{Kuznetsov1996}
to the case of rectangular matrices. While it is possible for $\gamma_{{\rm recycle}}^{N}$
to have a higher variance than $\gamma_{{\rm simple}}^{N}$, we show
that in many statistical scenarios it requires far fewer particles
to obtain a given variance, e.g.\ in the latent variable setting
described above. In particular, under weak assumptions, one needs
to take $N=\mathcal{O}(n)$ to control the relative variance of $\gamma_{{\rm recycle}}^{N}$
but one requires $N=\mathcal{O}(n^{2})$ to control the relative variance
of $\gamma_{{\rm simple}}^{N}$.

The paper is structured as follows. In Section~\ref{sec:The-associated-permanent}
we draw the connection with matrix permanents, define the proposed
estimator $\gamma_{{\rm recycle}}^{N}$, and analyze basic properties
such as unbiasedness and consistency. In Section~\ref{sec:Scaling-of-the}
we compare the behavior of $\gamma_{{\rm simple}}^{N}$ and $\gamma_{{\rm recycle}}^{N}$
as $n$ increases under various assumptions on the potential functions,
while in Section~\ref{sec:Latent-variable-models} we consider latent
variable models. Finally, Section~\ref{sec:Applications} provides
simulation studies showing the effect of using $\gamma_{{\rm recycle}}^{N}$
rather than $\gamma_{{\rm simple}}^{N}$ in pseudo-marginal Markov
chain Monte Carlo methods for estimating the parameters of a g-and-k
model, commonly used as a test application for approximate Bayesian
computation methodology, and of a Poisson-Beta model for single-cell
gene expression. In both cases we observe computational speedups of
the order of the number of data points $n$. Section~\ref{sec:Discussion}
provides a discussion and potential future works. All proofs are housed
in the appendix or the supplementary materials.

\section{The associated permanent and its approximation\label{sec:The-associated-permanent}}

For integers $i\leq j$ we denote $\left\llbracket i,j\right\rrbracket =\{i,\ldots,j\}$.
We adopt the convention that $\prod_{j\in\emptyset}G_{j}=1$, and
will occasionally use the notation $x_{p:q}=(x_{p},\ldots,x_{q})$
for $p,q\in\mathbb{N}$ with $p\leq q$. An alternative approximation
of $\gamma$ on the basis of the particles $\zeta=(\zeta_{1},\ldots,\zeta_{N})$
is obtained by first rewriting $\gamma$ in (\ref{eq:gamma}) as,
with $X_{1},\ldots,X_{p}$ independent $\mu$-distributed random variables,
\[
\gamma=\prod_{p=1}^{n}\mathbb{E}\left[G_{p}(X)\right]=\mathbb{E}\left[\prod_{p=1}^{n}G_{p}(X_{p})\right].
\]
Indeed, $\gamma_{{\rm biased}}^{N}$ is a V-statistic of order $n$
for $\gamma$, and the corresponding U-statistic for $\gamma$ is

\begin{equation}
\gamma_{{\rm perm}}^{N}:=\sum_{k\in P(N,n)}\left|P(N,n)\right|^{-1}\prod_{p=1}^{n}G_{p}(\zeta_{k_{p}}),\label{eq:perm_approx}
\end{equation}
where $P(N,n)=\{k\in\left\llbracket 1,N\right\rrbracket ^{n}:k_{i}=k_{j}\iff i=j\}$
is the set of $n$-permutations of $N$, whose cardinality is $\left|P(N,n)\right|=N!/(N-n)!$.
We observe that $\gamma_{{\rm perm}}^{N}$ is exactly $\left|P(N,n)\right|^{-1}$
times the permanent of the rectangular matrix $A$ \citep[see, e.g.,][p. 25]{ryser1963combinatorial}
with entries $A_{ij}=G_{i}(\zeta_{j})$ since then
\[
{\rm perm}(A)=\sum_{k\in P(N,n)}\prod_{p=1}^{n}A_{p,k_{p}}=\sum_{k\in P(N,n)}\prod_{p=1}^{n}G_{p}(\zeta_{k_{p}}).
\]
The approximation $\gamma_{{\rm perm}}^{N}$ is unbiased and consistent
since it is a U-statistic and moreover it is less variable than $\gamma_{{\rm simple}}^{N}$
in terms of the convex order \citep[see, e.g.,][Section 3.A]{shaked2007stochastic},
defined by $X\preceq_{{\rm cx}}Y$ if $\mathbb{E}\left[\phi(X)\right]\leq\mathbb{E}\left[\phi(Y)\right]$
for all convex functions $\phi:\mathbb{R}\rightarrow\mathbb{R}$ such
that the expectations are well-defined. Since $x\mapsto x$ and $x\mapsto-x$
are convex functions, convex-ordered random variables necessarily
have the same expectation, and since $x\mapsto x^{2}$ is convex,
$X\preceq_{{\rm cx}}Y$ implies ${\rm var}(X)\leq{\rm var}(Y)$. Convex-ordered
families of random variables also allow one to order the asymptotic
variances of associated pseudo-marginal Markov chains \citep[Theorem~10]{andrieu2016establishing}.
In order to express the second moments of $\gamma_{{\rm perm}}^{N}/\gamma$
and $\gamma_{{\rm recycle}}^{N}/\gamma$, we define the function $\psi_{N}:\left\llbracket 1,N\right\rrbracket ^{n}\rightarrow\mathbb{R}_{+}$
by
\begin{equation}
\psi_{N}(r):=\left\{ \prod_{p=1}^{n}\mu(\bar{G}_{p}\prod_{j:r_{j}=p}\bar{G}_{j})\right\} \left\{ \prod_{i=n+1}^{N}\mu(\prod_{j:r_{j}=i}\bar{G}_{j})\right\} ,\qquad r\in\left\llbracket 1,N\right\rrbracket ^{n},\label{eq:psiN_def}
\end{equation}
where $fg$ denotes pointwise product so that $\mu(fg)=\int f(x)g(x)\mu({\rm d}x)$.
We now state basic properties of $\gamma_{{\rm perm}}^{N}$, which
can be compared with Proposition~\ref{prop:simple_analysis}.
\begin{thm}
\label{thm:permanentisbetter}The following hold:
\begin{enumerate}
\item $\mathbb{E}[\gamma_{{\rm perm}}^{N}]=\gamma$ and $\gamma_{{\rm perm}}^{N}\preceq_{{\rm cx}}\gamma_{{\rm simple}}^{N}$.
\item $\gamma_{{\rm perm}}^{N}\overset{P}{\rightarrow}\gamma$ as $N\rightarrow\infty$.
\item The second moment of $\gamma_{{\rm perm}}^{N}/\gamma$ is, with $K\sim{\rm Uniform}(P(N,n))$,
\[
\mathbb{E}[(\gamma_{{\rm perm}}^{N}/\gamma)^{2}]=\mathbb{E}[\prod_{p=1}^{n}\bar{G}_{p}(\zeta_{p})\bar{G}_{p}(\zeta_{K_{p}})]=\mathbb{E}\left[\psi_{N}(K)\right]=\mathbb{E}[\prod_{p=1}^{n}\mu\big(\bar{G}_{p}\prod_{j:K_{j}=p}\bar{G}_{j}\big)].
\]
\item ${\rm var}(\gamma_{{\rm perm}}^{N})$ is finite and ${\rm var}(\gamma_{{\rm perm}}^{N})\rightarrow0$
as $N\rightarrow\infty$ if and only if (\ref{eq:second_moment_condition})
holds.
\end{enumerate}
\end{thm}
This suggests that $\gamma_{{\rm perm}}^{N}$ is a superior approximation
of $\gamma$ in comparison to $\gamma_{{\rm simple}}^{N}$. However,
computing $\gamma_{{\rm perm}}^{N}$ is equivalent to computing the
permanent of a rectangular matrix, which has no known polynomial-time
algorithm. In fact, computing the permanent of a square matrix is
\#P-hard \citep{valiant1979complexity}. Using an extension of the
importance sampling estimator of the permanent of a square matrix
due to \citet{Kuznetsov1996}, we define the following unbiased approximation
of $\gamma_{{\rm perm}}^{N}$ and hence $\gamma$,
\begin{equation}
\gamma_{{\rm recycle}}^{N}:=\prod_{p=1}^{n}\frac{1}{N-p+1}\sum_{j=1}^{N}G_{p}(\zeta_{j})\mathbb{I}\left(j\notin\{K_{1},\ldots,K_{p-1}\}\right),\label{eq:perm_approx_tilde}
\end{equation}
where $K:=(K_{1},\ldots,K_{n})$ is a $\left\llbracket 1,N\right\rrbracket ^{n}$-valued
random variable whose distribution given $\zeta$ is defined by the
sequence of conditional probabilities
\begin{equation}
\mathbb{P}\left[K_{p}=i\mid\zeta,K_{1},\ldots,K_{p-1}\right]\propto G_{p}(\zeta_{i})\mathbb{I}\left(i\notin\{K_{1},\ldots,K_{p-1}\}\right).\label{eq:Kp_conditionals_approx}
\end{equation}
In (\ref{eq:Kp_conditionals_approx}) we take $G_{p}(\zeta_{i})/\sum_{j=1}^{N}G_{p}(\zeta_{j})\mathbb{I}\left(j\notin\{K_{1},\ldots,K_{p-1}\}\right)$
to be $1$ when the denominator is equal to $0$. That is, when the
denominator is $0$, then $K_{p}\mid(\zeta,K_{1},\ldots,K_{p-1})\sim{\rm Uniform}(\left\llbracket 1,N\right\rrbracket \setminus\{K_{1},\ldots,K_{p-1}\})$.
The choice of the conditional distribution of $K_{p}$ when $\sum_{j=1}^{N}G_{p}(\zeta_{j})\mathbb{I}\left(j\notin\{K_{1},\ldots,K_{p-1}\}\right)=0$
is in some sense arbitrary, as in any case $\gamma_{{\rm recycle}}^{N}=0$
whenever this happens. We now state basic properties of $\gamma_{{\rm recycle}}^{N}$,
which can be compared with Theorem~\ref{thm:permanentisbetter}.

\begin{algorithm}

\caption{Computing $\gamma_{{\rm recycle}}^{N}$\label{alg:recycle}}
\begin{enumerate}
\item Sample $\zeta_{1},\ldots,\zeta_{N}\sim\mu$ independently, and set
$Z_{0}\leftarrow1$.
\item For $p=1,\ldots,n$:
\begin{enumerate}
\item Set $Z_{p}\leftarrow Z_{p-1}\sum_{j=1}^{N}G_{p}(\zeta_{j})\mathbb{I}\left(j\notin\{K_{1},\ldots,K_{p-1}\}\right)/(N-p+1)$.
\item Sample $K_{p}\mid(K_{1},\ldots,K_{p-1})$ according to (\ref{eq:Kp_conditionals_approx}).
\end{enumerate}
\item Set $\gamma_{{\rm recycle}}^{N}\leftarrow Z_{n}$.
\end{enumerate}
\end{algorithm}
\begin{thm}
\label{thm:approx_consistency}The following hold:
\begin{enumerate}
\item $\mathbb{E}\left[\gamma_{{\rm recycle}}^{N}\mid\zeta\right]=\gamma_{{\rm perm}}^{N}$,
$\mathbb{E}\left[\gamma_{{\rm recycle}}^{N}\right]=\gamma$ and $\gamma_{{\rm perm}}^{N}\preceq_{{\rm cx}}\gamma_{{\rm recycle}}^{N}$.
\item $\gamma_{{\rm recycle}}^{N}\overset{P}{\rightarrow}\gamma$ as $N\rightarrow\infty$.
\item Let $S=(S_{1},\ldots,S_{n})$ be a vector of independent random variables
with $S_{p}\sim{\rm Uniform}(\left\llbracket p,N\right\rrbracket )$
for $p\in\left\llbracket 1,n\right\rrbracket $. The second moment
of $\gamma_{{\rm recycle}}^{N}/\gamma$ is
\[
\mathbb{E}\left[(\gamma_{{\rm recycle}}^{N}/\gamma)^{2}\right]=\mathbb{E}[\prod_{p=1}^{n}\bar{G}_{p}(\zeta_{p})\bar{G}_{p}(\zeta_{S_{p}})]=\mathbb{E}\left[\psi_{N}(S)\right].
\]
\item ${\rm var}(\gamma_{{\rm recycle}}^{N})$ is finite and ${\rm var}(\gamma_{{\rm recycle}}^{N})\rightarrow0$
as $N\rightarrow\infty$ if and only if
\begin{equation}
\max_{p\in\left\llbracket 1,n\right\rrbracket ,\:B\subseteq\left\llbracket 1,p\right\rrbracket }\mu(G_{p}\prod_{j\in B}G_{j})<\infty.\label{eq:var.finite.cond}
\end{equation}
\end{enumerate}
\end{thm}
\begin{cor}
\label{cor:moments_for_consistency}If $\max_{p\in\left\llbracket 1,n\right\rrbracket }\mu(G_{p}^{n+1})<\infty$
then ${\rm var}(\gamma_{{\rm recycle}}^{N})\rightarrow0$ as $N\rightarrow\infty$.
\end{cor}
\begin{rem}
We observe that $\mathbb{E}[(\gamma_{{\rm perm}}^{N}/\gamma)^{2}]=\mathbb{E}\left[\psi_{N}(K)\right]\leq\mathbb{E}\left[\psi_{N}(S)\right]=\mathbb{E}[(\gamma_{{\rm recycle}}^{N}/\gamma)^{2}]$,
where $K$ and $S$ are defined in the statements of Theorems~\ref{thm:permanentisbetter}
and~\ref{thm:approx_consistency}, respectively. 
\end{rem}
\begin{rem}
\label{rem:moments_approximations}While (\ref{eq:second_moment_condition})
is sufficient for $\gamma_{{\rm perm}}^{N}$ and $\gamma_{{\rm simple}}^{N}$
to have finite variance converging to $0$ as $N\rightarrow\infty$,
this is not sufficient in general for $\gamma_{{\rm recycle}}^{N}$,
which requires (\ref{eq:var.finite.cond}) instead.
\end{rem}
The estimator $\gamma_{{\rm simple}}^{N}$ uses only $N/n$ out of
$N$ particles to estimate each expectation in the product; in contrast
$\gamma_{{\rm recycle}}^{N}$ uses $N-p$ particles for the $p$th
expectation $\mu(G_{p})$. In this sense, the latter recycles most
of the particles for each term, and we therefore refer to $\gamma_{{\rm recycle}}^{N}$
as the recycled estimator in the sequel. While Remark~\ref{rem:moments_approximations}
implies that it is not possible for ${\rm var}(\gamma_{{\rm recycle}}^{N})\leq{\rm var}(\gamma_{{\rm simple}}^{N})$
in general, we show in the coming section that ${\rm var}(\gamma_{{\rm recycle}}^{N})$
can be orders of magnitude smaller than ${\rm var}(\gamma_{{\rm simple}}^{N})$
in many statistical settings.

\section{Scaling of the number of particles with $n$\label{sec:Scaling-of-the}}

We investigate the variance of $\gamma_{{\rm recycle}}^{N}$ in comparison
to $\gamma_{{\rm simple}}^{N}$ in the large $n$ regime. In particular,
we show that only $N=\mathcal{O}(n)$ particles are required to control
the relative variance of $\gamma_{{\rm recycle}}^{N}$ in some scenarios
in which $N=\mathcal{O}(n^{2})$ particles are required to control
the relative variance of $\gamma_{{\rm simple}}^{N}$. We also show
that this cannot always be true, in some situations $N=\mathcal{O}(n^{2})$
is a lower bound on the number of particles required to control the
relative variance of $\gamma_{{\rm perm}}^{N}$, and therefore $\gamma_{{\rm recycle}}^{N}$.
To simplify the presentation, we define $c_{p}:=\mu(\bar{G}_{p}^{2})-1$,
for $p\in\left\llbracket 1,n\right\rrbracket $. We will occasionally
make reference to the following assumption when considering the large
$n$ regime
\begin{equation}
1<\inf_{p\geq1}c_{p}\leq\sup_{p\geq1}c_{p}<\infty.\label{eq:cs_assume}
\end{equation}
We begin by observing that from Proposition~\ref{prop:simple_analysis},
if (\ref{eq:cs_assume}) holds and $M=\left\lceil \alpha n^{\beta}\right\rceil $
then the second moment of $\gamma_{{\rm simple}}^{N}/\gamma$ is bounded
above as $n\rightarrow\infty$ if and only if $\alpha>0$ and $\beta\geq1$.
Since $N=Mn$, this implies that to stabilize the relative variance
of $\gamma_{{\rm simple}}^{N}$ in the large $n$ regime one must
take $N=\mathcal{O}(n^{2})$.

The second moment of $\gamma_{{\rm recycle}}^{N}/\gamma$ is more
complex to analyze because it involves interactions between different
potential functions. The following results consider three particular
situations. The first is a favorable scenario, in which it follows
that if $N=\alpha n$ for $\alpha>1$ and (\ref{eq:cs_assume}) holds
with $c=\sup_{p\geq1}c_{p}$, then $\mathbb{E}[(\gamma_{{\rm recycle}}^{N}/\gamma)^{2}]\leq\exp\left(c/\left[\alpha-1\right]\right)$. 
\begin{prop}
\label{prop:independent_Gs}Assume $G_{1}(X),\ldots,G_{n}(X)$ are
mutually independent when $X\sim\mu$. Then
\begin{equation}
\mathbb{E}\left[\left(\gamma_{{\rm recycle}}^{N}/\gamma\right)^{2}\right]=\prod_{p=1}^{n}\left[1+c_{p}/(N-p+1)\right].\label{eq:sec_mom_ind_Gs}
\end{equation}
\end{prop}
The second scenario is also favorable: if the potential functions
are negatively correlated in a specific sense, $\gamma_{{\rm recycle}}^{N}=\gamma_{{\rm perm}}^{N}$
and again it is sufficient to take $N=\mathcal{O}(n)$ to control
the second moment of $\gamma_{{\rm recycle}}^{N}/\gamma$.
\begin{prop}
\label{prop:neg_correlatedGs}Assume that for any distinct $p,q\in\left\llbracket 1,n\right\rrbracket $,
$G_{p}(X)G_{q}(X)=0$ almost surely. Then $\gamma_{{\rm recycle}}^{N}=\gamma_{{\rm perm}}^{N}$
almost surely, and $\mathbb{E}\left[(\gamma_{{\rm recycle}}^{N}/\gamma)^{2}\right]\leq\prod_{p=1}^{n}\left[1+c_{p}/(N-p+1)\right]$.
\end{prop}
The third scenario is not favorable, corresponding to the case where
the potential functions are identical. At least a quadratic in $n$
number of particles is required to control $\mathbb{E}[(\gamma_{{\rm recycle}}^{N}/\gamma)^{2}]$
in this setting. Loosely speaking, positive correlations between $G_{p}(X)$
and $G_{q}(X)$ tend to increase the second moment of $\gamma_{{\rm recycle}}^{N}/\gamma$,
while correlations have no effect on the second moment of $\gamma_{{\rm simple}}^{N}/\gamma$.
Nevertheless, we show in Proposition~\ref{prop:identicalGs_momentcond}
that when the moments of $\bar{G}_{1}(X)$ increase no more quickly
than those associated to rescaled Bernoulli random variables, $\gamma_{{\rm recycle}}^{N}$
has a smaller variance than $\gamma_{{\rm simple}}^{N}$. Hence the
recycled estimator may be useful, even if not orders of magnitude
better, in some applications involving the same potential functions,
e.g.\ the Poisson estimator of \citet{beskos2006exact}, based on
general methods by \citet{bhanot1985bosonic} and \citet{wagner1988unbiased}. 
\begin{prop}
\label{prop:pos_correlatedGs}Assume $G_{1}=\cdots=G_{n}$. Then $\mathbb{E}\left[(\gamma_{{\rm recycle}}^{N}/\gamma)^{2}\right]\geq\mathbb{E}\left[(\gamma_{{\rm perm}}^{N}/\gamma)^{2}\right]\geq(1+c_{1})^{n^{2}/N}$.
\end{prop}
\begin{prop}
\label{prop:identicalGs_momentcond}Let $G_{1}=\cdots=G_{n}$ and
assume that $\mu(\bar{G}_{1}^{\ell})\leq\mu(\bar{G_{1}}^{2})^{\ell-1}$
for $\ell\in\left\llbracket 3,n\right\rrbracket $. Then $\mathbb{E}\left[(\gamma_{{\rm recycle}}^{N}/\gamma)^{2}\right]\leq\mathbb{E}\left[(\gamma_{{\rm simple}}^{N}/\gamma)^{2}\right]$.
\end{prop}
Our final general result is motivated by approximate Bayesian computation
applications, in which it is often the case that the potential functions
are indicator functions. In this case, it is also true that $\gamma_{{\rm recycle}}^{N}$
has a smaller variance than $\gamma_{{\rm simple}}^{N}$.
\begin{prop}
\label{prop:indicator_better}Let $A_{1},\ldots,A_{n}\in\mathcal{E}$
satisfy $\mu(A_{p})>0$ for $p\in\left\llbracket 1,n\right\rrbracket $.
Let $G_{p}:=\mathbb{I}{}_{A_{p}}$ for $p\in\left\llbracket 1,n\right\rrbracket $.
Then $\mathbb{E}\left[(\gamma_{{\rm recycle}}^{N}/\gamma)^{2}\right]\leq\mathbb{E}\left[(\gamma_{{\rm simple}}^{N}/\gamma)^{2}\right]$.
\end{prop}
\begin{rem}
An alternative approximation of $\gamma$ can be obtained by sampling
a number of permutations $\sigma_{1},\ldots,\sigma_{q}$ of $\left\llbracket 1,n\right\rrbracket $
and calculating $\gamma_{{\rm simple}}^{N}$ using each permutation.
That is, if we define
\[
\gamma_{{\rm simple}}^{N}(\sigma)=\prod_{p=1}^{n}\frac{1}{M}\sum_{i=1}^{M}G_{p}(\zeta_{\sigma((p-1)M+i)}),
\]
then $q^{-1}\sum_{i=1}^{q}\gamma_{{\rm simple}}^{N}(\sigma_{i})$
is also an approximation of $\gamma_{{\rm perm}}^{N}$ and hence $\gamma$.
This strategy does not scale well with $n$, however. For example,
if (\ref{eq:cs_assume}) holds and $M=\left\lceil \alpha n^{\beta}\right\rceil $
with $\alpha>0$ and $\beta\in(0,1)$ we require $q$ to grow exponentially
with $n$ to stabilize the relative variance of $q^{-1}\sum_{i=1}^{q}\gamma_{{\rm simple}}^{N}(\sigma_{i})$.
The crucial observation to establish this is that for non-negative
random variables $W_{1},\ldots,W_{q}$ with identical means and variances,
we have $\mathbb{E}\left[(q^{-1}\sum_{i=1}^{q}W_{i})^{2}\right]\geq q^{-1}\mathbb{E}\left[W_{1}^{2}\right]$,
and the argument follows from Proposition~\ref{prop:simple_analysis}.
\end{rem}

\section{Latent variable models\label{sec:Latent-variable-models}}

The assumption of mutual independence in Proposition~\ref{prop:independent_Gs}
is very strong in statistical settings. However, we show now that
in latent variable models the expected second moment of $\gamma_{{\rm recycle}}^{N}/\gamma$
is very similar to (\ref{eq:sec_mom_ind_Gs}), where the expectation
$\mathsf{E}$ is w.r.t.\ the law of the observations $Y_{1},\ldots,Y_{n}\overset{i.i.d.}{\sim}\nu$.
For the remainder of this section, we denote by $\bar{G}_{1}$ the
random function $x\mapsto g(x,Y_{1})/\nu(Y_{1})$ for $Y_{1}\sim\nu$.
We begin by verifying (\ref{eq:var.finite.cond}) for latent variable
models under a finite expected second moment condition for $\bar{G}_{1}(X)$
when $X\sim\mu$. This condition has appeared in the literature in
a variety of places, see e.g. \citet{breiman1985estimating}, \citet{buja1990remarks},
\citet{schervish1992convergence}, \citet{liu1995covariance} and
\citet{khare2011spectral}.
\begin{prop}
\label{prop:finite_var_as}In the setting of Example~\ref{example:latent-var-model},
assume that $\mathsf{E}\left[\mu(\bar{G}_{1}^{2})\right]<\infty$.
If $Y_{1},\ldots,Y_{n}\overset{i.i.d.}{\sim}\nu_{0}$ where $\nu_{0}$
is absolutely continuous with respect to $\nu$, then (\ref{eq:var.finite.cond})
holds almost surely.
\end{prop}
The following Theorem is our main result in terms of applicability
to statistical scenarios. It suggests that when considering the expected
second moment of $\gamma_{{\rm recycle}}^{N}/\gamma$, it is as if
the random variables $G_{1}(X),\ldots,G_{n}(X)$ are ``mutually independent
on average'', and allows easy comparison with the corresponding expected
second moment of $\gamma_{{\rm simple}}^{N}/\gamma$.
\begin{thm}
\label{thm:dg_expected_second_moment}In the setting of Example~\ref{example:latent-var-model},
and letting $\mathsf{E}$ denoting expectation w.r.t. $Y_{1},\ldots,Y_{n}$,
\[
\mathsf{E}\left[\mathbb{E}\left[\left(\gamma_{{\rm recycle}}^{N}/\gamma\right)^{2}\right]\right]=\prod_{p=1}^{n}\left[1+C/(N-p+1)\right],\qquad C=\mathsf{E}\left[\mu(\bar{G}_{1}^{2})\right]-1.
\]
\end{thm}
\begin{rem}
In the setting of Example~\ref{example:latent-var-model}, it is
straightforward to obtain from Proposition~\ref{prop:simple_analysis}
that $\mathsf{E}[\mathbb{E}[(\gamma_{{\rm simple}}^{N}/\gamma)^{2}]]=\left(1+C/M\right)^{n}$,
where $C$ is as in Theorem~\ref{thm:dg_expected_second_moment}.
Hence, one requires $N=\left\lceil \alpha n\right\rceil $ for $\alpha>1$
to control the expected relative variance of $\gamma_{{\rm recycle}}^{N}$
but one requires $M=\mathcal{O}(n)$ and hence $N=\mathcal{O}(n^{2})$
to control the expected relative variance of $\gamma_{{\rm simple}}^{N}$
when $1<C<\infty$. In addition, it is clear that $\mathsf{E}[\mathbb{E}[(\gamma_{{\rm recycle}}^{N}/\gamma)^{2}]]<\mathsf{E}[\mathbb{E}[(\gamma_{{\rm simple}}^{N}/\gamma)^{2}]]$
for any $N$ that is an integer multiple of $n>1$.
\end{rem}
\begin{rem}
The condition $\mathsf{E}\left[\mu(\bar{G}_{1}^{2})\right]<\infty$
is not very strong, but is not always satisfied. For example, if $\mu$
is ${\rm Uniform}(0,1)$ and $g(x,\cdot)$ is ${\rm Uniform}(0,x)$
for each $x\in(0,1)$ then simple calculations show that $\mathsf{E}\left[\mu(\bar{G}_{1}^{2})\right]=\infty$. 
\end{rem}

\section{Applications\label{sec:Applications}}

We consider Bayesian inference in two latent variable model applications,
employing $\gamma_{{\rm recycle}}^{N}$ or $\gamma_{{\rm simple}}^{N}$
to approximate $L(\theta)$ in a pseudo-marginal version of a random-walk
Metropolis Markov chain. General guidelines for tuning the value of
$N$ in such chains have been proposed by \citet{Doucet07032015}
and \citet{sherlock2015}, who suggest that one should choose $N$
such that the relative variance of the estimator is roughly $2$.
While the relative variance typically varies with $\theta$, if the
posterior distribution for $\theta$ is reasonably concentrated near
the true parameter $\theta_{0}$, in practice one can often choose
$N$ so that the estimator has a relative variance of $2$ at some
point close to $\theta_{0}$. In both applications below, following
\citet{roberts2001optimal}, we tune the proposal for the random-walk
Metropolis algorithm using a shorter run of the Markov chain. Specifically,
we choose the proposal density
\[
q(\theta,\theta')=\mathcal{N}(\theta';\theta,d^{-1/2}2.38\hat{\Sigma}),
\]
where $\hat{\Sigma}$ is the estimated covariance matrix of the posterior
distribution and $\theta\in\mathbb{R}^{d}$. All computations were
performed in the R programming language with C++ code accessed via
the `Rcpp' package \citep{rcpp}. Effective sample sizes were computed
using the `mcmcse' package \citep{mcmcse}.

Using $\gamma_{{\rm perm}}^{N}$ instead of $\gamma_{{\rm simple}}^{N}$
to approximate each $L(\theta)$ in a pseudo-marginal Markov chain
can only decrease the asymptotic variance of ergodic averages of functions
$\varphi$ with ${\rm var}_{\pi}(\varphi)<\infty$. This is a consequence
of \citet[Theorem~10]{andrieu2016establishing} and Theorem~\ref{thm:permanentisbetter}.
Using $\gamma_{{\rm recycle}}^{N}$ does not have the same guarantee
in general, but Theorem~\ref{thm:dg_expected_second_moment} suggests
that if the estimators perform similarly for a set of $\theta$ with
large posterior mass, then this should result in greatly improved
performance over $\gamma_{{\rm simple}}^{N}$ for large $n$.

\subsection{Approximate Bayesian computation: g-and-k model}

Approximate Bayesian computation (ABC) is a branch of simulation-based
inference used when the likelihood function cannot be evaluated pointwise
but one can simulate from the model for any value of the statistical
parameter. While there are a number of variants, in general the methodology
involves comparing a summary statistic associated with the observed
data with summary statistics associated with pseudo-data simulated
using different parameter values \citep[see][for a recent review]{Marin2012}.
When the data are modelled as $n$ observations of i.i.d.\  random
variables with distribution $\mu$, it is commonplace to summarize
the data using some fixed-dimensional summary statistic independent
of $n$, for computational rather than statistical reasons. This summarization,
or dimension reduction, can in principle involve little loss of information
about the parameters \textemdash{} in exponential families sufficient
statistics of fixed dimension exist and could be computed or approximated
\textemdash{} but in practice this is not always easy to achieve.
An alternative approach that we adopt here is to eschew dimension
reduction altogether and treat the model as a standard latent variable
model, essentially using the noisy ABC methodology of \citet{fearnhead2012constructing}.
This may be viewed as an alternative to the construction of summaries
using the Wasserstein distance recently proposed by \citet{bernton2017inference}.
One possible outcome of this is that less data may be required to
achieve a given degree of posterior concentration; a theoretical treatment
of this is beyond the scope of this paper.

The $g$-and-$k$ distribution has been used as an example application
for ABC methods since \citet{allingham2009bayesian}. The distribution
is parameterized by $\theta=(A,B,g,k)$ and a sample $X$ from this
distribution can be expressed as 
\[
X=A+B\left\{ 1+c\frac{1-\exp(-gZ)}{1+\exp(-gZ)}\right\} \left\{ 1+Z^{2}\right\} ^{k}Z,
\]
where $Z\sim\mathcal{N}(0,1)$ is a standard normal random variable
and we fix the value $c=4/5$. We consider here observations $y_{1},\ldots,y_{n}$
of $n$ independent random variables $Y_{i}=X_{i}+\epsilon U_{i}$
where $X_{i}\overset{i.i.d.}{\sim}g\text{-and-}k(\theta_{0})$ and
$U_{i}\overset{i.i.d.}{\sim}{\rm Uniform}(-\epsilon,\epsilon)$ with
$\theta_{0}=(3,1,2,0.5)$ and $\epsilon=1/5$. The values of $c$
and $\theta_{0}$ follow \citet{allingham2009bayesian}. We let $\mu_{\theta}$
denote the distribution of $X\sim g\text{-and-}k(\theta)$, and define
$G_{p}=\mathbb{I}_{(y_{p}-\epsilon,y_{p}+\epsilon)}$, so that $\gamma(\theta)=\prod_{p=1}^{n}\mu_{\theta}(G_{p})$
is equivalent to the likelihood $L(\theta)$ associated with $\theta$.
We take $n=100$ and, following \citet{allingham2009bayesian}, we
put independent ${\rm Uniform}(0,10)$ priors on each component of
$\theta$. In order to have a relative variance of $\gamma_{{\rm recycle}}^{N}(\theta_{0})$
of roughly $2$, it was sufficient to take $N=100n=10^{4}$ whereas
for $\gamma_{{\rm simple}}^{N}(\theta_{0})$ we required $N=100n^{2}=10^{6}$.
Using both estimators resulted in very similar Markov chains, but
the computational cost of using the simple estimator was over $30$
times greater; it took $25.6$ hours to simulate a simple chain of
length $10^{6}$ and $8.4$ hours to simulate a recycled chain of
length $10^{7}$. It would have taken over $10$ days to simulate
a simple chain of length $10^{7}$. Figure~\ref{fig:gk_densities}
shows posterior density estimates associated with the recycled chain;
effective sample sizes for each component were above $80,000$. In
this example, simulating from $\mu$ is approximately $1200$ times
more expensive than evaluating a potential function. Finally, we observe
that the posterior distribution for $\theta$ places most of its mass
near $\theta_{0}$ despite using $n=100$; in contrast \citet{allingham2009bayesian}
used $n=10^{5}$ and Figure~\ref{fig:gk_densities} shows more concentration
overall and better identification of the $g$ parameter than their
Figure~3. This suggests that this type of latent variable approach
may be preferable to dimension-reducing summaries in some i.i.d. ABC
models.

\begin{figure}
\subfloat[$A$]{

\includegraphics[scale=0.8]{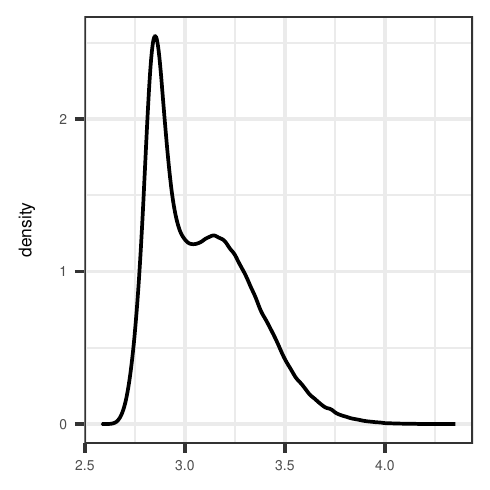}}\subfloat[$B$]{\includegraphics[scale=0.8]{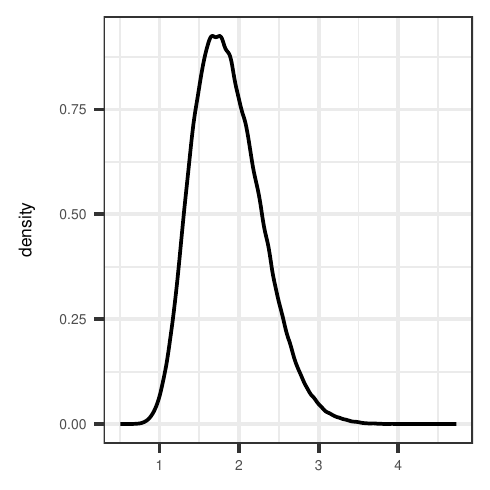}}\subfloat[$g$]{\includegraphics[scale=0.8]{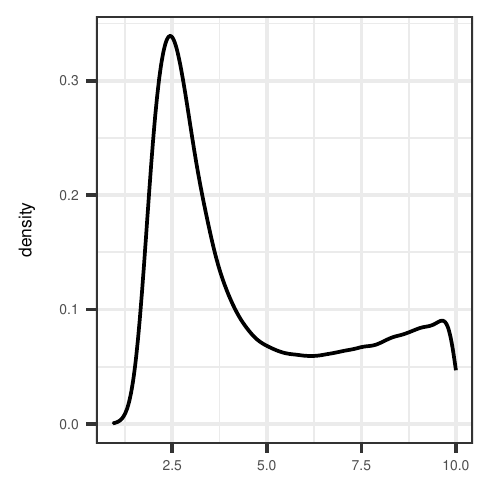}}\subfloat[$k$]{\includegraphics[scale=0.8]{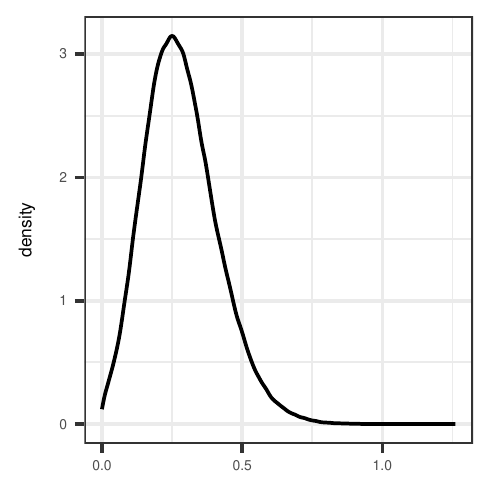}}\caption{\label{fig:gk_densities}Posterior density estimates for $\theta$
for the g-and-k model.}
\end{figure}

\subsection{Poisson-Beta model for gene expression}

\citet{peccoud1995markovian} proposed a continuous-time birth-and-death
process in a random environment to model single-cell gene expression
levels; this model enjoys strong experimental support \citep{delmans2016discrete}.
Letting $V=(V_{t})_{t\geq0}$ denote the $\mathbb{Z}_{+}$-valued
process counting the amount of transcribed mRNA and $W=(W_{t})_{t\geq0}$
denote the $\{0,1\}$-valued process indicating whether the gene is
inactive or active, the Markov process $X=(X_{t})_{t\geq0}=(V,W)$
is described by, with $X_{t}=(v,w)$,
\[
\mathbb{P}\left[X_{t+dt}=(v',w')\right]=\begin{cases}
v\,dt+o(dt) & (v',w')=(v-1,w),\\{}
[k_{{\rm on}}(1-w)+k_{{\rm off}}w]\,dt+o(dt) & (v',w')=(v,1-w),\\
\lambda\,dt+o(dt) & (v',w',w)=(v+1,1,1),\\
0 & \text{otherwise}.
\end{cases}
\]
The statistical parameters $k_{{\rm on}}$, $k_{{\rm off}}$ and $\lambda$
are respectively the rates at which: the gene switches from inactive
to active, the gene switches from active to inactive, and, mRNA is
transcribed when the gene is active. The rate of mRNA degradation
is assumed here to be $1$. \citet{peccoud1995markovian} derive the
probability generating function of the stationary distribution of
this process, from which one obtains the probability mass function
of the stationary marginal distribution of $V$
\[
\mu(v)=\frac{\lambda^{v}\int_{0}^{1}t^{k_{{\rm on}}+v-1}(1-t)^{k_{{\rm off}}-1}e^{-\lambda t}\,dt}{v!\:{\rm Beta}(k_{{\rm on}},k_{{\rm off}})},\qquad v\in\mathbb{Z}_{+}.
\]
As observed by, e.g., \citet{kim2013inferring}, straightforward calculations
provide that this is equivalent to the probability mass function of
$X\sim\text{Poisson-Beta}(\lambda,k_{{\rm on}},k_{{\rm off}})$, defined
hierarchically by $X\mid S\sim{\rm Poisson}(\lambda S)$ where $S\sim{\rm Beta}(k_{{\rm on}},k_{{\rm off}})$.
This model was also mentioned in \citet{wills2013single}, who described
the Poisson-Beta model directly.

In an experiment, one might observe in the stationary regime mRNA
counts with noise for $n$ independent cells, which can therefore
be modelled as $n$ independent random variables $Y_{1},\ldots,Y_{n}$
with distribution $Y_{i}=X_{i}+\sigma Z_{i}$ where $X_{i}\overset{{\rm ind}}{\sim}\text{Poisson-Beta}(\lambda,k_{{\rm on}},k_{{\rm off}})$
and $Z_{1},\ldots,Z_{n}$ are independent standard normal random variables.
Hence, this can be viewed as a latent variable model with $\theta=(\lambda,k_{{\rm on}},k_{{\rm off}})$,
$\mu_{\theta}=\text{Poisson-Beta}(\theta)$ and $G_{p}(x)=\mathcal{N}(y_{p};x,\sigma^{2})$,
and likelihood function values $L(\theta)$ are exactly of the form
described in Remark~\ref{rem:likelihood-function-lvarmod}. We simulated
data $y_{1},\ldots,y_{n}$ with $n=1000$, $\theta_{0}=(500,2,8)$
and $\sigma=5$, and proceeded to conduct Bayesian inference via pseudo-marginal
MCMC with independent exponential priors on $\lambda$, $k_{{\rm on}}$
and $k_{{\rm off}}$ with means of $1000$, $10$ and $10$, respectively.
In order to have a relative variance of $\gamma_{{\rm recycle}}^{N}(\theta_{0})$
of roughly $2$, it was sufficient to take $N=20n=2\times10^{4}$
whereas for $\gamma_{{\rm simple}}^{N}(\theta_{0})$ we required $N=20n^{2}=2\times10^{7}$.
Using both estimators resulted in very similar Markov chains, but
the computational cost of using the simple estimator was approximately
$600$ times greater; it took $12.25$ hours to simulate a simple
chain of length $10^{4}$ and $6.15$ hours to simulate a recycled
chain of length $3\times10^{6}$. It would have taken over $153$
days to simulate a simple chain of length $3\times10^{6}$. Figure~\ref{fig:bp_densities}
shows posterior density estimates associated with the recycled chain;
effective sample sizes for each component were above $23,000$. In
this example, simulating from $\mu$ is approximately $2700$ times
more expensive than evaluating a potential function.

\begin{figure}
\subfloat[$\lambda$]{\includegraphics[scale=1.1]{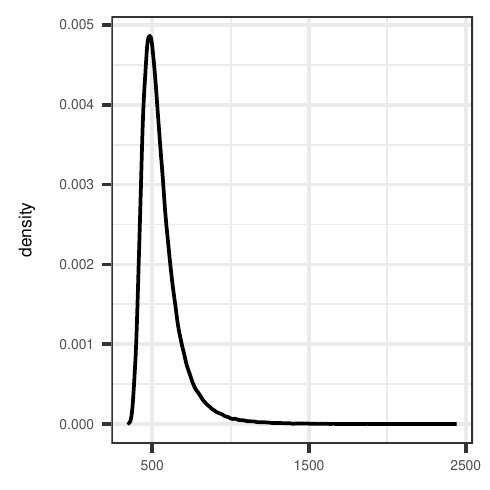}}\subfloat[$k_{{\rm on}}$]{\includegraphics[scale=1.1]{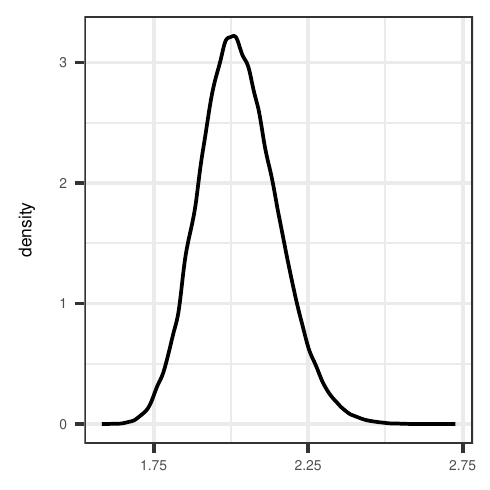}}\subfloat[$k_{{\rm off}}$]{\includegraphics[scale=1.1]{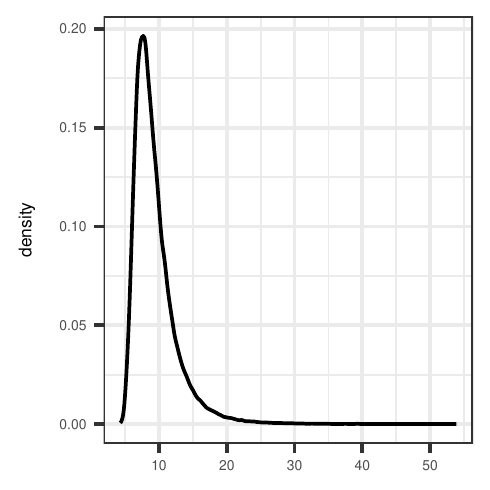}}\caption{\label{fig:bp_densities}Posterior density estimates for $\theta$
for the Poisson-Beta model.}
\end{figure}

\section{Discussion\label{sec:Discussion}}

We have proposed an unbiased estimator of a product of expectations
that involves using, or recycling, most of the random variables simulated.
This results in considerable decreases in the computational time required
to approximate such a product accurately when the number of terms
in the product, $n$, is large and the computational cost of simulating
random variables is significantly larger than that of evaluating functions
involved in the terms. In latent variable models, we have shown that
the number of samples $N$ required for a given relative variance
is proportional to $n$, while for a simple estimator one requires
$N$ to be proportional to $n^{2}$. We have demonstrated that the
use of the recycled estimator proposed here successfully reduces computational
time for Bayesian inference using pseudo-marginal Markov chain Monte
Carlo from days or months to hours in some situations. It would be
interesting to see if the methodology could be combined with the correlated
particle filter methodology of \citet{deligiannidis2015correlated}
to bring further improvements.

Relating the results on numbers of samples required to common notions
of asymptotic time complexity requires some care. For a given relative
variance in the setting of Theorem~\ref{thm:dg_expected_second_moment},
one can choose $\alpha$ such that the following approximately holds.
The number of samples required for the recycled estimator is $\alpha n$
and the number of function evaluations is slightly less than $\alpha n^{2}$
while for the simple estimator we require $\alpha n^{2}$ samples
and $\alpha n^{2}$ function evaluations. The computational time for
the recycled estimator can be expressed as $\alpha n(c_{s}+nc_{g}+nc_{r})$
where $c_{s}$ is the cost of sampling from $\mu$, $c_{g}$ the cost
of evaluating a potential function, and $c_{r}$ is the problem-independent
time per particle associated with Step 2(b) in Algorithm~\ref{alg:recycle}.
For the simple estimator, the computational time is $\alpha n^{2}(c_{s}+c_{g})$
and so the recycled estimator is $(c_{g}+c_{s})/(c_{g}+c_{r})$ times
faster than the simple estimator as $n\rightarrow\infty$, so that
the improvement depends almost entirely on the relative differences
between $c_{s}$, $c_{g}$ and $c_{r}$. In both our applications,
$c_{s}$ is over a thousand times larger than $c_{g}$.

There are alternative unbiased approximations of the permanent of
a rectangular matrix that could be used in place of the approach due
to \citet{Kuznetsov1996}. In particular, it is straightforward to
extend the algorithm of \citet{kou2009approximating} to the rectangular
case. However, we have found the corresponding approximation to be
orders of magnitude worse than \citeauthor{Kuznetsov1996}'s for the
rectangular matrices used here. This is due to the fact that \citeauthor{kou2009approximating}'s
algorithm is specifically designed to overcome deficiencies of \citeauthor{Kuznetsov1996}'s
algorithm for square matrices by emphasizing the importance of large
values in relation to others in the same column. In the rectangular
case, \citeauthor{kou2009approximating}'s algorithm overcompensates
in this regard as this consideration is less important. There are
also much more computationally expensive approximations of the permanent,
such as \citet{wang2016monte}, which may be useful in situations
where simulations are very expensive in comparison to function evaluations.

It would be of interest to obtain accurate, general lower bounds for
the second moment of $\gamma_{{\rm perm}}^{N}$ to complement the
upper bounds for $\gamma_{{\rm recycle}}^{N}$, particularly in the
setting of Example~\ref{example:latent-var-model} to complement
Theorem~\ref{thm:dg_expected_second_moment}. We have been able to
show that in the setting of Proposition~\ref{prop:independent_Gs},
there we have $\mathbb{E}\left[(\gamma_{{\rm perm}}^{N}/\gamma)^{2}\right]\geq\prod_{p=1}^{n}\left[1+c_{p}/N\right]$
but the argument did not extend naturally to the setting of Theorem~\ref{thm:dg_expected_second_moment}.
Finally, it is straightforward to define $\gamma_{{\rm recycle}}^{N}$
alternatively by choosing a permutation $\sigma$ of $\left\llbracket 1,n\right\rrbracket $
according to any distribution and re-ordering the $G_{1},\ldots,G_{n}$
as $G_{\sigma(1)},\ldots,G_{\sigma(n)}$. The corresponding condition
to (\ref{eq:var.finite.cond}), if the distribution for $\sigma$
places mass on every possible permutation of $\left\llbracket 1,n\right\rrbracket $,
is then $\max_{p\in\left\llbracket 1,n\right\rrbracket ,\:B\subseteq\left\llbracket 1,n\right\rrbracket }\mu(G_{p}\prod_{j\in B}G_{j})<\infty$.

It is straightforward to define a recycled estimator of a product
of $n$ expectations, each with respect to a different distribution.
Letting $\mu_{1},\ldots,\mu_{n}$ denote the distributions, one can
define a common dominating probability distribution $\tilde{\mu}$
and take, for each $p\in\left\llbracket 1,n\right\rrbracket $, $\tilde{G}{}_{p}=G_{p}\cdot{\rm d}\mu_{p}/{\rm d}\tilde{\mu}$
so that $\tilde{\mu}(\tilde{G}_{p})=\mu_{p}(G_{p})$. That is, one
can re-express the product of expectations as a product of expectations
all with respect to $\tilde{\mu}$. The results of Sections~\ref{sec:The-associated-permanent}
and~\ref{sec:Scaling-of-the} then apply, and the recycled estimator
could be very useful when $\tilde{\mu}(\tilde{G}_{p}^{2})/\tilde{\mu}(\tilde{G}_{p})^{2}$
is not too large for any $p\in\left\llbracket 1,n\right\rrbracket $.
One can also compare variances of the simple estimator with the recycled
estimator through Proposition~\ref{prop:simple_analysis} and Theorem~\ref{thm:approx_consistency},
even though in this case the simple estimator would use blocks of
independent random variables from $\mu_{1},\ldots,\mu_{n}$ whereas
the recycled estimator would use independent $\tilde{\mu}$-distributed
random variables.

\appendix

\section*{Acknowledgments}

Some of this research was undertaken while all three authors were
at the University of Warwick. The authors acknowledge helpful discussions
and comments from Christophe Andrieu, Christopher Jennison and Matti
Vihola. AL supported by The Alan Turing Institute under the EPSRC
grant EP/N510129/1. GZ supported in part by an EPSRC Doctoral Prize
fellowship and by the European Research Council (ERC) through StG
``N-BNP'' 306406.

\section*{Appendices}

\section{Proof of Theorem~\ref{thm:permanentisbetter}}

The following lemma provides a sufficient condition for two random
variables to be convex-ordered, which is useful for our analysis.
\begin{lem}
\label{lem:suff_cx_cond}Let $X$ and $Y$ be random variables such
that $\mathbb{E}[X]$ and $\mathbb{E}[Y]$ are well-defined. Then
$X\preceq_{{\rm cx}}Y$ if there exists a probability space with random
variables $X'$ and $Y'$ equal in distribution respectively to $X$
and $Y$, and an additional random variable $Z'$ such that $\mathbb{E}\left[Y'\mid Z'\right]=X'$
almost surely.
\end{lem}
\begin{proof}
For convex $\phi$, Jensen's inequality provides
\[
\mathbb{E}\left[\phi(Y')\right]=\mathbb{E}\left[\mathbb{E}\left[\phi(Y')\mid Z'\right]\right]\geq\mathbb{E}\left[\phi(\mathbb{E}\left[Y'\mid X'\right])\right]=\mathbb{E}\left[\phi(X')\right],
\]
so $\mathbb{E}\left[\phi(Y)\right]\geq\mathbb{E}\left[\phi(X)\right]$.
\end{proof}
\begin{rem}
Lemma~\ref{lem:suff_cx_cond} is related to the deeper and well-known
Strassen Representation Theorem \citep[Theorem~8]{strassen1965existence}
which states, with the notation of Lemma~\ref{lem:suff_cx_cond},
that $X\preceq_{{\rm cx}}Y$ if and only if $\mathbb{E}\left[Y'\mid X'\right]=X'$
almost surely. In fact, $\mathbb{E}\left[Y'\mid Z'\right]=X'$ implies
$\mathbb{E}\left[Y'\mid X',Z'\right]=X'$, which is equivalent to
a conditional convex order between $X$ and $Y$, and which implies
$X\preceq_{{\rm cx}}Y$ \citep[cf.][Proposition~2.1(ii)]{leskela2014conditional}.
\end{rem}
\begin{lem}
\label{lem:psi_expression}Let $R$ be a $\left\llbracket 1,N\right\rrbracket ^{n}$-valued
random variable independent of $\zeta$. Then
\[
\mathbb{E}\left[\prod_{p=1}^{n}\bar{G}_{p}(\zeta_{p})\bar{G}_{p}(\zeta_{R_{p}})\right]=\mathbb{E}\left[\psi_{N}(R)\right].
\]
\end{lem}
\begin{proof}
From the law of total expectation $\mathbb{E}\left[\prod_{p=1}^{n}\bar{G}_{p}(\zeta_{p})\bar{G}_{p}(\zeta_{R_{p}})\right]=\mathbb{E}\left[\mathbb{E}\left[\prod_{p=1}^{n}\bar{G}_{p}(\zeta_{p})\bar{G}_{p}(\zeta_{R_{p}})\mid R\right]\right]$,
and by collecting together terms involving the independent $\zeta_{1},\ldots,\zeta_{N}$,
for any $r\in\left\llbracket 1,N\right\rrbracket ^{n}$,
\begin{align*}
\mathbb{E}\left[\prod_{p=1}^{n}\bar{G}_{p}(\zeta_{p})\bar{G}_{p}(\zeta_{r_{p}})\mid R=r\right] & =\mathbb{E}\left[\left\{ \prod_{p=1}^{n}\bar{G}_{p}(\zeta_{p})\prod_{j:r_{j}=p}\bar{G}_{j}(\zeta_{p})\right\} \left\{ \prod_{i=n+1}^{N}\prod_{j:r_{j}=i}\bar{G}_{j}(\zeta_{i})\right\} \right]\\
 & =\left\{ \prod_{p=1}^{n}\mathbb{E}\left[\bar{G}_{p}(\zeta_{p})\prod_{j:r_{j}=p}\bar{G}_{j}(\zeta_{p})\right]\right\} \left\{ \prod_{i=n+1}^{N}\mathbb{E}\left[\prod_{j:r_{j}=i}\bar{G}_{j}(\zeta_{i})\right]\right\} \\
 & =\left\{ \prod_{p=1}^{n}\mu\left(\bar{G}_{p}\prod_{j:r_{j}=p}\bar{G}_{j}\right)\right\} \left\{ \prod_{i=n+1}^{N}\mu\left(\prod_{j:r_{j}=i}\bar{G}_{j}\right)\right\} .\qedhere
\end{align*}
\end{proof}
\begin{proof}[Proof of Theorem~\ref{thm:permanentisbetter}]
The lack-of-bias follows from the fact that that $\gamma_{{\rm perm}}^{N}$
is a U-statistic. Recalling that $N=Mn$, let $C(N,n)$ be the set
of partitions of $\left\llbracket 1,N\right\rrbracket $ such that
each $A\in C(N,n)$ consists of elements of cardinality exactly $M$
denoted $A_{1},\ldots,A_{n}$. Then for any $A\in C(N,n)$, $\gamma_{A}^{N}:=\prod_{p=1}^{n}\frac{1}{N}\sum_{i\in A_{p}}G_{p}(\zeta_{i})$
is equal in distribution to $\gamma_{{\rm simple}}^{N}$. Moreover,
by symmetry we have
\[
\gamma_{{\rm perm}}^{N}=\sum_{k\in P(N,n)}\frac{(N-n)!}{(N)!}\prod_{p=1}^{n}G_{p}(\zeta_{k_{p}})=\sum_{A\in C(N,n)}\frac{1}{|C(N,n)|}\gamma_{A}^{N}=\mathbb{E}\left[\gamma_{S}^{N}\mid\zeta\right],
\]
where in the last equality $S\sim{\rm Uniform}(C(N,n))$. Hence, Lemma~\ref{lem:suff_cx_cond}
implies $\gamma_{{\rm perm}}^{N}\preceq_{{\rm cx}}\gamma_{{\rm simple}}^{N}$.
The consistency follows from the Strong Law of Large Numbers for U-statistics
\citep{hoeffding1961strong}. For the second moment of $\gamma_{{\rm perm}}^{N}/\gamma$,
from the definition (\ref{eq:perm_approx}) of $\gamma_{{\rm perm}}^{N}$,
we have
\begin{align*}
\mathbb{E}\left[\left(\gamma_{{\rm perm}}^{N}/\gamma\right)^{2}\right] & =\left|P(N,n)\right|^{-2}\mathbb{E}\left[\sum_{k\in P(N,n)}\sum_{r\in P(N,n)}\prod_{p=1}^{n}\bar{G}_{p}(\zeta_{k_{p}})\bar{G}_{p}(\zeta_{r_{p}})\right]\\
 & =\left|P(N,n)\right|^{-1}\sum_{k\in P(N,n)}\mathbb{E}\left[\prod_{p=1}^{n}\bar{G}_{p}(\zeta_{p})\bar{G}_{p}(\zeta_{k_{p}})\right]=\mathbb{E}\left[\prod_{p=1}^{n}\bar{G}_{p}(\zeta_{p})\bar{G}_{p}(\zeta_{K_{p}})\right],
\end{align*}
the second equality following from the exchangeability of $\zeta$.
Applying Lemma~\ref{lem:psi_expression} we obtain
\begin{align*}
\mathbb{E}\left[\prod_{p=1}^{n}\bar{G}_{p}(\zeta_{p})\bar{G}_{p}(\zeta_{K_{p}})\right] & =\mathbb{E}\left[\left\{ \prod_{p=1}^{n}\mu\left(\bar{G}_{p}\prod_{j:K_{j}=p}\bar{G}_{j}\right)\right\} \left\{ \prod_{i=n+1}^{N}\mu\left(\prod_{j:K_{j}=i}\bar{G}_{j}\right)\right\} \right]\\
 & =\mathbb{E}\left[\prod_{p=1}^{n}\mu\left(\bar{G}_{p}\prod_{j:K_{j}=p}\bar{G}_{j}\right)\right],
\end{align*}
the second equality following since $\mu(\bar{G}_{p})=1$ for each
$p\in\left\llbracket 1,n\right\rrbracket $ and $K$ being $P(N,n)$-valued
implies $\max_{i\in\left\llbracket n+1,N\right\rrbracket }\left|\{j:K_{j}=i\}\right|\leq1$.
For the last part, assume first that (\ref{eq:second_moment_condition})
holds. We observe then that $\max_{k\in P(N,n),p\in\left\llbracket 1,n\right\rrbracket }\mu(\bar{G}_{p}\prod_{j:k_{j}=p}\bar{G}_{j})<\infty$
since for any $k\in P(N,n)$ and $p\in\left\llbracket 1,n\right\rrbracket $,
$\left|\{j:k_{j}=p\}\right|\leq1$ and the Cauchy\textendash Schwarz
inequality implies that $\mu(\bar{G}_{p}\bar{G}_{q})^{2}\leq\mu(\bar{G}_{p}^{2})\mu(\bar{G}_{q}^{2})<\infty$
for any $p,q\in\left\llbracket 1,n\right\rrbracket $. It follows
that ${\rm var}(\gamma_{{\rm perm}}^{N})<\infty$. We can write
\[
\mathbb{E}\left[(\gamma_{{\rm perm}}^{N}/\gamma)^{2}\right]=\mathbb{E}\left[\psi_{N}(K)\right]=\mathbb{P}(A)+\mathbb{P}(A^{\complement})\mathbb{E}\left[\psi_{N}(K)\mid A^{\complement}\right],
\]
where $A=\{K\in P(N,n)\::\:\min\{K_{1},\ldots,K_{n}\}>n\}$. We observe
that $\mathbb{P}(A)=\prod_{i=1}^{n}\frac{N-n-i+1}{N-i+1}\rightarrow1$
for fixed $n$ as $N\rightarrow\infty$, and for $k\in P(N,n)$ we
have
\[
\mathbb{E}\left[\psi_{N}(K)\mid A^{\complement}\right]\leq\max_{k\in P(N,n)}\psi_{N}(k)=\max_{k\in P(N,n)}\prod_{p=1}^{n}\mu(\bar{G}_{p}\prod_{j:k_{j}=p}\bar{G}_{j})\leq\prod_{p=1}^{n}\max_{q\in\left\llbracket 1,n\right\rrbracket }\mu\left(\bar{G}_{p}\bar{G}_{q}\right),
\]
where the R.H.S. is finite and independent of $N$. Hence, $\mathbb{E}\left[(\gamma_{{\rm perm}}^{N}/\gamma)^{2}\right]\rightarrow1$
and ${\rm var}(\gamma_{{\rm perm}}^{N})\rightarrow0$, as desired.
Finally, assume that $\mu(G_{q}^{2})=\infty$ for some $q\in\left\llbracket 1,n\right\rrbracket $
and let $k=(1,\ldots,n)$. We have then
\[
\mathbb{E}\left[(\gamma_{{\rm perm}}^{N}/\gamma)^{2}\right]=\mathbb{E}\left[\psi_{N}(K)\right]\geq\left|P(N,n)\right|^{-1}\psi_{N}(k),
\]
and since $\psi_{N}(k)=\prod_{p=1}^{n}\mu(\bar{G}_{p}^{2})=\infty$,
we conclude that ${\rm var}(\gamma_{{\rm perm}}^{N})=\infty$.
\end{proof}

\section{Proof of Theorem~\ref{thm:approx_consistency}}
\begin{proof}[Proof of Theorem~\ref{thm:approx_consistency}]
This is a consequence of Lemmas~\ref{lem:lack-biased-convex-order-permapprox},~\ref{lem:consistency_permapprox},~\ref{lem:second_moment_permapprox},
and~\ref{lem:variance_characterization} below.
\end{proof}
\begin{lem}
\label{lem:lack-biased-convex-order-permapprox}$\mathbb{E}\left[\gamma_{{\rm recycle}}^{N}\mid\zeta\right]=\gamma_{{\rm perm}}^{N}$,
$\mathbb{E}\left[\gamma_{{\rm recycle}}^{N}\right]=\gamma$ and $\gamma_{{\rm perm}}^{N}\preceq_{{\rm cx}}\gamma_{{\rm recycle}}^{N}$.
\end{lem}
\begin{proof}
Combining (\ref{eq:perm_approx_tilde}) and (\ref{eq:Kp_conditionals_approx})
we obtain,
\[
\mathbb{E}\left[\gamma_{{\rm recycle}}^{N}\mid\zeta\right]=\left|P(N,n)\right|^{-1}\sum_{k\in P(N,n)}\prod_{p=1}^{n}G_{p}(\zeta_{k_{p}})=\gamma_{{\rm perm}}^{N}.
\]
It follows that $\mathbb{E}\left[\gamma_{{\rm recycle}}^{N}\right]=\mathbb{E}\left[\gamma_{{\rm perm}}^{N}\right]=\gamma$
and by Lemma~\ref{lem:suff_cx_cond}, $\gamma_{{\rm perm}}^{N}\preceq_{{\rm cx}}\gamma_{{\rm recycle}}^{N}$.
\end{proof}
\begin{lem}
\label{lem:consistency_permapprox}$\gamma_{{\rm recycle}}^{N}\overset{P}{\rightarrow}\gamma$
as $N\rightarrow\infty$.
\end{lem}
\begin{proof}
Let $\tilde{Z}_{p,N}:=(N-p+1)^{-1}\sum_{j=1}^{N}G_{p}(\zeta_{j})\mathbb{I}(j\notin\{K_{1},\dots,K_{p-1}\})$
for $p\in\left\llbracket 1,n\right\rrbracket $, so that $\gamma_{{\rm recycle}}^{N}=\sum_{p=1}^{n}\tilde{Z}_{p,N}$.
We will show $\tilde{Z}_{p,N}\stackrel{P}{\rightarrow}\mu(G_{p})$
for every $p\in\left\llbracket 1,n\right\rrbracket $ and deduce,
by Slutsky's Theorem, that $\gamma_{{\rm recycle}}^{N}\stackrel{P}{\rightarrow}\gamma$.
Since $(N-p+1)/N\rightarrow1$ as $N\rightarrow\infty$, $\tilde{Z}_{p,N}\stackrel{P}{\rightarrow}\mu(G_{p})$
is equivalent to $\tilde{Z}_{p,N}(N-p+1)/N\stackrel{P}{\rightarrow}\mu(G_{p})$.
Since
\[
\frac{N-p+1}{N}\tilde{Z}_{p,N}=\frac{\sum_{j=1}^{N}G_{p}(\zeta_{j})}{N}-\frac{\sum_{i=1}^{p-1}G_{p}(\zeta_{K_{i}})}{N},
\]
and $N^{-1}\sum_{j=1}^{N}G_{p}(\zeta_{j})\stackrel{P}{\rightarrow}\mu(G_{p})$
as $N\to\infty$ by the weak Law of Large Numbers, it suffices to
show that $N^{-1}\sum_{i=1}^{p-1}G_{p}(\zeta_{K_{i}})\stackrel{P}{\rightarrow}0$
as $N\to\infty$. Since $\sum_{i=1}^{p-1}G_{p}(\zeta_{K_{i}})=0$
for $p=1$, it suffices to show the result for $p\geq2$. For any
$\epsilon>0$,
\begin{align*}
\mathbb{P}\left[\frac{\sum_{i=1}^{p-1}G_{p}(\zeta_{K_{i}})}{N}\geq\epsilon\right] & \leq\mathbb{P}\left[\frac{(p-1)}{N}\max_{j\in\left\llbracket 1,N\right\rrbracket }G_{p}(\zeta_{j})\geq\epsilon\right]\\
 & \leq\sum_{j=1}^{N}\mathbb{P}\left[G_{p}(\zeta_{j})\geq\frac{\epsilon N}{p-1}\right]=\frac{p-1}{\epsilon}\frac{N\epsilon}{p-1}\mathbb{P}\left[G_{p}(\zeta_{1})\geq\frac{\epsilon N}{p-1}\right]\\
 & \leq\frac{p-1}{\epsilon}\mathbb{E}\left[G_{p}(\zeta_{1})\mathbb{I}\left\{ G_{p}(\zeta_{1})\geq\frac{\epsilon N}{p-1}\right\} \right].
\end{align*}
Since $\mathbb{E}\left[G_{p}(\zeta_{1})\right]=\mu(G_{p})<\infty$,
the last term converges to $0$ as $N\rightarrow\infty$ since it
is the tail of a convergent integral, and we conclude.
\end{proof}
\begin{lem}
\label{lem:second_moment_permapprox}The second moment of $\gamma_{{\rm recycle}}^{N}/\gamma$
can be expressed as
\[
\mathbb{E}\left[(\gamma_{{\rm recycle}}^{N}/\gamma)^{2}\right]=\mathbb{E}\left[\prod_{p=1}^{n}\bar{G}_{p}(\zeta_{p})\bar{G}_{p}(\zeta_{S_{p}})\right]=\mathbb{E}\left[\psi_{N}(S)\right],
\]
where $S=(S_{1},\ldots,S_{n})$ is a vector of independent random
variables with $S_{p}\sim{\rm Uniform}(\left\llbracket p,N\right\rrbracket )$
for $p\in\left\llbracket 1,n\right\rrbracket $.
\end{lem}
\begin{proof}
We obtain from (\ref{eq:perm_approx_tilde}) and (\ref{eq:Kp_conditionals_approx})
that
\[
\mathbb{E}\left[(\gamma_{{\rm recycle}}^{N}/\gamma)^{2}\biggm|\zeta\right]=\left|P(N,n)\right|^{-2}\sum_{k\in P(N,n)}\prod_{p=1}^{n}\left[\bar{G}_{p}(\zeta_{k_{p}})\sum_{j\in\left\llbracket 1,N\right\rrbracket \setminus\{k_{1},\ldots,k_{p-1}\}}\bar{G}_{p}(\zeta_{j})\right].
\]
From exchangeability of $\zeta$ and the law of total expectation,
we then obtain
\[
\mathbb{E}\left[(\gamma_{{\rm recycle}}^{N}/\gamma)^{2}\right]=\left|P(N,n)\right|^{-1}\mathbb{E}\left[\prod_{p=1}^{n}\bar{G}_{p}(\zeta_{p})\sum_{j=p}^{N}\bar{G}_{p}(\zeta_{j})\right]=\mathbb{E}\left[\prod_{p=1}^{n}\bar{G}_{p}(\zeta_{p})\bar{G}_{p}(\zeta_{S_{p}})\right],
\]
and conclude by applying Lemma~\ref{lem:psi_expression}.
\end{proof}
\begin{lem}
\label{lem:variance_characterization}${\rm var}(\gamma_{{\rm recycle}}^{N})$
is finite and ${\rm var}(\gamma_{{\rm recycle}}^{N})\rightarrow0$
as $N\rightarrow\infty$ if and only if (\ref{eq:var.finite.cond})
holds.
\end{lem}
\begin{proof}
Assume (\ref{eq:var.finite.cond}) holds. From Lemma~\ref{lem:second_moment_permapprox}
and (\ref{eq:psiN_def}), 
\[
\mathbb{E}\left[(\gamma_{{\rm recycle}}^{N}/\gamma)^{2}\right]=\mathbb{E}\left[\psi_{N}(S)\right]=\mathbb{P}(A)+\mathbb{P}(A^{\complement})\mathbb{E}\left[\psi_{N}(S)\mid A^{\complement}\right],
\]
where $A=\{S\::\:\min_{i}S_{i}>n\;\text{and}\;S_{i}\neq S_{j},\quad\forall i\neq j\}$.
We observe that $\mathbb{P}(A)=\prod_{i=1}^{n}\frac{N-n-i+1}{N-i+1}\rightarrow1$
for fixed $n$ as $N\rightarrow\infty$. We consider first the term
$\prod_{p=1}^{n}\mu(\bar{G}_{p}\prod_{j:S_{j}=p}\bar{G}_{j})$ in
(\ref{eq:psiN_def}). For each $p\in\left\llbracket 1,n\right\rrbracket $,
we can write $\mu(\bar{G}_{p}\prod_{j:S_{j}=p}\bar{G}_{j})=\mu(\bar{G}_{p}\prod_{j\in B}\bar{G}_{j})$
for $B=\{j\in\left\llbracket 1,n\right\rrbracket :S_{j}=p\}$ and
$B\subseteq\left\llbracket 1,p\right\rrbracket $ from the definition
of $S$. Hence,
\[
\prod_{p=1}^{n}\mu(\bar{G}_{p}\prod_{j:S_{j}=p}\bar{G}_{j})\leq\prod_{p=1}^{n}\max_{B\subseteq\left\llbracket 1,p\right\rrbracket }\mu(\bar{G}_{p}\prod_{j\in B}\bar{G}_{j})<\infty,
\]
by (\ref{eq:var.finite.cond}). Now, the term $\prod_{i=n+1}^{N}\mu(\prod_{j:S_{j}=i}\bar{G}_{j})$
is a product of at most $n$ terms different from $1$, each of which
can be written as $\mu(\prod_{j:S_{j}=i}\bar{G}_{j})=\mu(\prod_{j\in\tilde{B}}\bar{G}_{j})$
for some $\tilde{B}\subseteq\left\llbracket 1,n\right\rrbracket $
and hence as $\mu(\prod_{j:S_{j}=i}\bar{G}_{j})=\mu(\bar{G}_{p}\prod_{j\in B}\bar{G}_{j})$
for $p=\max(\tilde{B})$ and $B\subseteq\left\llbracket 1,p\right\rrbracket $.
Therefore,
\[
\prod_{i=n+1}^{N}\mu\left(\prod_{j:S_{j}=i}\bar{G}_{j}\right)\leq\left\{ \max_{p\in\left\llbracket 1,n\right\rrbracket ,B\subseteq\left\llbracket 1,p\right\rrbracket }\mu\left(\bar{G}_{p}\prod_{j\in B}\bar{G}_{j}\right)\right\} ^{n}<\infty
\]
again by (\ref{eq:var.finite.cond}). It follows that 
\[
\mathbb{E}\left[\psi_{N}(S)\mid A^{\complement}\right]\leq\left\{ \max_{p\in\left\llbracket 1,n\right\rrbracket ,B\subseteq\left\llbracket 1,p\right\rrbracket }\mu(\bar{G}_{p}\prod_{j\in B}\bar{G}_{j})\right\} ^{2n},
\]
where the R.H.S. is finite and independent of $N$. Hence, $\mathbb{E}\left[(\gamma_{{\rm recycle}}^{N}/\gamma)^{2}\right]\rightarrow1$
and ${\rm var}(\gamma_{{\rm recycle}}^{N})\rightarrow0$, as desired.

Suppose now that (\ref{eq:var.finite.cond}) does not hold. Then there
exists $q\in\left\llbracket 1,n\right\rrbracket $ and $B\in\left\llbracket 1,q\right\rrbracket $
such that $\mu(G_{q}\prod_{j\in B}G_{j})=\infty$. We treat separately
the case where $q\in B$ and when $q\notin B$. If $q\in B$, then
let $s\in\left\llbracket 1,N\right\rrbracket ^{n}$ be defined by
$s_{p}=q$ for $p\in B$ and $s_{p}=p$ for $p\in\left\llbracket 1,n\right\rrbracket \setminus B$.
It can then be checked that
\[
\psi_{N}(s)=\mu\left(\bar{G}_{q}\prod_{j\in B}\bar{G}_{j}\right)\left[\prod_{p\in B}\mu(\bar{G}_{p})\right]\left[\prod_{p\in\left\llbracket 1,n\right\rrbracket \setminus B}\mu(\bar{G}_{p}^{2})\right],
\]
so $\psi_{N}(s)=\infty$ because $\mu(\bar{G}_{q}\prod_{j\in B}\bar{G}_{j})=\infty$
and the other terms are non-zero. Since $\mathbb{P}(S=s)>0$, where
$S$ is defined in Lemma~\ref{lem:second_moment_permapprox}, it
follows that $\mathbb{E}\left[(\gamma_{{\rm recycle}}^{N}/\gamma)^{2}\right]=\mathbb{E}\left[\psi_{N}(S)\right]=\infty$.
If instead $q\notin B$, then let $s\in\left\llbracket 1,N\right\rrbracket ^{n}$
be defined by $s_{p}=q$ for all $p\in B$, $s_{q}=r$ for some $r\in B\setminus\{q\}$
and $s_{p}=p$ for $p\in\left\llbracket 1,n\right\rrbracket \setminus B$.
It can then be checked that
\[
\psi_{N}(s)=\mu\left(\bar{G}_{q}\prod_{j\in B}\bar{G}_{j}\right)\mu\left(\bar{G}_{r}\bar{G}_{q}\right)\left[\prod_{p\in B\setminus\{r\}}\mu\left(\bar{G}_{p}\right)\right]\left[\prod_{p\in\left\llbracket 1,n\right\rrbracket \setminus\{r\}}\mu\left(\bar{G}_{p}^{2}\right)\right],
\]
so $\psi_{N}(s)=\infty$ because $\mu(\bar{G}_{q}\prod_{j\in B}\bar{G}_{j})=\infty$
and the other terms are non-zero. Since $\mathbb{P}(S=s)>0$, it follows
as before that $\mathbb{E}\left[(\gamma_{{\rm recycle}}^{N}/\gamma)^{2}\right]=\mathbb{E}\left[\psi_{N}(S)\right]=\infty$.
\end{proof}
\begin{proof}[Proof of Corollary~\ref{cor:moments_for_consistency}]
By Theorem~\ref{thm:approx_consistency}, it suffices to show that
$\max_{p\in\left\llbracket 1,n\right\rrbracket }\mu(G_{p}^{n+1})<\infty$
implies (\ref{eq:var.finite.cond}). Consider an arbitrary term $\mu\big(G_{p}\prod_{j\in B}G_{j}\big)$
in (\ref{eq:var.finite.cond}). The generalized Hölder inequality
\citep[see, e.g.,][p. 67]{kufner1977function} implies that 
\[
\mu(G_{p}\prod_{j\in B}G_{j})\leq\mu(G_{p}^{|B|+1})^{1/(|B|+1)}\prod_{j\in S}\mu(G_{j}^{|B|+1})^{1/(|B|+1)}.
\]
Since $|B|+1\leq(n+1)$ we have $\mu(G_{j}^{|B|+1})^{1/(|B|+1)}\leq\mu(G_{j}^{n+1})^{1/(n+1)}$
by applying the Hölder inequality to $G_{j}$ and the constant random
variable equal to 1. Therefore 
\[
\mu(G_{p}\prod_{j\in B}G_{j})\leq\mu(G_{p}^{n+1})^{1/(n+1)}\prod_{j\in B}\mu(G_{j}^{n+1})^{1/(n+1)}\leq\max_{p\in\left\llbracket 1,n\right\rrbracket }\mu(G_{p}^{n+1}),
\]
and so $\max_{p\in\left\llbracket 1,n\right\rrbracket }\mu(G_{p}^{n+1})<\infty$
implies (\ref{eq:var.finite.cond}).
\end{proof}

\section{Proofs of Propositions~\ref{prop:independent_Gs}\textendash \ref{prop:pos_correlatedGs}}
\begin{proof}[Proof of Proposition~\ref{prop:independent_Gs}]
From the assumption it follows that $\mu\big(\bar{G}_{p}\prod_{j:s_{j}=p}\bar{G}_{j}\big)=\mu\big(\bar{G}_{p}^{2}\big)^{\mathbb{I}(s_{p}=p)}$
and $\mu\big(\prod_{j:s_{j}=p}\bar{G}_{j}\big)=1$. Therefore by Lemma~\ref{lem:second_moment_permapprox},
and with $S=(S_{1},\ldots,S_{n})$ a vector of independent random
variables with $S_{p}\sim{\rm Uniform}(\left\llbracket p,N\right\rrbracket )$
for $p\in\left\llbracket 1,n\right\rrbracket $,
\[
\mathbb{E}\left[(\gamma_{{\rm recycle}}^{N}/\gamma)^{2}\right]=\mathbb{E}\left[\psi_{N}(S)\right]=\mathbb{E}\left[\prod_{p=1}^{n}\mu\big(\bar{G}_{p}^{2}\big)^{\mathbb{I}(S_{p}=p)}\right]=\prod_{p=1}^{n}\mathbb{E}\left[\mu\big(\bar{G}_{p}^{2}\big)^{\mathbb{I}(S_{p}=p)}\right].
\]
We conclude from $\mathbb{P}(S_{p}=p)=1/(N-p+1)$ and $c_{p}=\mu\big(\bar{G}_{p}^{2})-1$.
\end{proof}
\begin{proof}[Proof of Proposition~\ref{prop:neg_correlatedGs}]
Almost surely, we have $\sum_{i=1}^{N}G_{p}(\zeta_{i})G_{q}(\zeta_{i})=0$
so
\[
\sum_{j=1}^{N}G_{p}(\zeta_{j})\mathbb{I}\left(j\notin\{K_{1},\ldots,K_{p-1}\}\right)=\sum_{j=1}^{N}G_{p}(\zeta_{j}),
\]
and hence
\[
\gamma_{{\rm recycle}}^{N}=\prod_{p=1}^{n}\frac{1}{N-p+1}\sum_{j=1}^{N}G_{p}(\zeta_{j})=\frac{1}{\left|P(N,n)\right|}\sum_{k\in P(N,n)}\prod_{p=1}^{n}G_{p}(\zeta_{k_{p}})=\gamma_{{\rm perm}}^{N}.
\]
The inequality holds because $\mu\big(\bar{G}_{p}\prod_{j:s_{j}=p}\bar{G}_{j}\big)\leq\mu\big(\bar{G}_{p}^{2}\big)^{\mathbb{I}(s_{p}=p)}$
and $\mu\big(\prod_{j:s_{j}=p}\bar{G}_{j}\big)\leq1$, following the
same reasoning as in the proof of Proposition~\ref{prop:independent_Gs}.
\end{proof}
\begin{proof}[Proof of Proposition~\ref{prop:pos_correlatedGs}]
Theorem~\ref{thm:approx_consistency} implies that $\mathbb{E}[(\gamma_{{\rm recycle}}^{N}/\gamma)^{2}]\geq\mathbb{E}[(\gamma_{{\rm perm}}^{N}/\gamma)^{2}]$,
and Theorem~\ref{thm:permanentisbetter} implies that $\mathbb{E}[(\gamma_{{\rm perm}}^{N}/\gamma)^{2}]=\mathbb{E}\left[\psi(K)\right]=\mathbb{E}[(1+c_{1})^{Z}]$,
where $K\sim{\rm Uniform}(P(N,n))$ and $Z=\sum_{p=1}^{n}\mathbb{I}(K_{p}\leq n)$.
Jensen's inequality then provides $\mathbb{E}[(1+c_{1})^{Z}]\geq(1+c_{1})^{\mathbb{E}[Z]}$,
and we conclude since $\mathbb{E}[Z]=\sum_{p=1}^{n}\mathbb{E}[\mathbb{I}\{K_{p}\leq n\}]=\sum_{p=1}^{n}\mathbb{P}[K_{p}\leq n]=n^{2}/N$.
\end{proof}

\section{Proof of Proposition~\ref{prop:finite_var_as} and Theorem~\ref{thm:dg_expected_second_moment}}
\begin{proof}[Proof of Proposition~\ref{prop:finite_var_as}]
First assume that $Y_{1},\ldots,Y_{n}\overset{i.i.d.}{\sim}\nu$,
and let $p\in\left\llbracket 1,n\right\rrbracket $ and $B\subseteq\left\llbracket 1,p\right\rrbracket $
be arbitrary, and $\mathsf{E}$ denote expectation w.r.t.\ the law
of $Y_{1},\ldots,Y_{n}$. From $\bar{G}_{p}(x)=g(x,y_{p})/\nu(y_{p})$
for $p\in\left\llbracket 1,n\right\rrbracket $ we obtain 
\begin{align*}
\mathsf{E}\left[\mu(\bar{G}_{p}\prod_{j\in B}\bar{G}_{j})\right] & =\int\nu^{\otimes p}(y_{1:p})\mu({\rm d}x)\frac{g(x,y_{p})}{\nu(y_{p})}\prod_{j\in B}\frac{g(x,y_{j})}{\nu(y_{j})}{\rm d}y_{1:p}\\
 & =\int\mu({\rm d}x)\frac{g(x,y_{p})^{1+\mathbb{I}(p\in B)}}{\nu(y_{p})^{\mathbb{I}(p\in B)}}{\rm d}y_{p},
\end{align*}
which is equal to $1$ if $p\notin B$ and to $\mathsf{E}\left[\mu(\bar{G}_{1}^{2})\right]$
if $p\in B$. It follows that $\mathsf{E}\left[\mu(\bar{G}_{p}\prod_{j\in B}\bar{G}_{j})\right]<\infty$
and so $\mu(\bar{G}_{p}\prod_{j\in B}\bar{G}_{j})$ is finite almost
surely. The extension to $Y_{1},\ldots,Y_{n}\overset{i.i.d.}{\sim}\nu_{0}$
is immediate.
\end{proof}
To prove Theorem~\ref{thm:dg_expected_second_moment} we first need
the following Lemma.
\begin{lem}
\label{lem:expected_muprod}Let $r\in\left\llbracket 1,n\right\rrbracket ^{\ell}$
satisfy $r_{1}<\cdots<r_{\ell}$. Then $\mathsf{E}\left[\mu(\bar{G}_{r_{1}}\cdots\bar{G}_{r_{\ell}})\mid Y_{r_{1}+1},\ldots,Y_{n}\right]=\mu(\bar{G}_{r_{2}}\cdots\bar{G}_{r_{\ell}})$.
\end{lem}
\begin{proof}
Without loss of generality, let $r_{i}=i$ for $i\in\llbracket1,\ell\rrbracket$.
Then, 
\begin{align*}
\mathsf{E}\left[\mu(\bar{G}_{1}\cdots\bar{G}_{\ell})\mid Y_{2},\ldots,Y_{n}\right] & =\int_{\mathsf{Y}}\int_{E}\frac{g(x,y_{1})}{\nu(y_{1})}\left[\prod_{k=2}^{\ell}\bar{G}_{k}(x)\right]\mu({\rm d}x)\nu(y_{1}){\rm d}y_{1}\\
 & =\int_{E}\int_{\mathsf{Y}}g(x,y_{1}){\rm d}y_{1}\left[\prod_{k=2}^{\ell}\bar{G}_{k}(x)\right]\mu({\rm d}x),
\end{align*}
and we conclude since $\int_{\mathsf{Y}}g(x,y_{1}){\rm d}y_{1}=1$.
\end{proof}
\begin{proof}[Proof of Theorem~\ref{thm:dg_expected_second_moment}]
We can write $\mathsf{E}\left[\mathbb{E}\left[(\gamma_{{\rm recycle}}^{N}/\gamma)^{2}\right]\right]=\mathbb{E}\left[\mathsf{E}\left[\psi_{N}(S,Y)\right]\right]$,
where
\[
\psi_{N}(s,y):=\left\{ \prod_{p=1}^{n}\mu\left(\bar{G}_{p}\prod_{j:s_{j}=p}\bar{G}_{j}\right)\right\} \left\{ \prod_{i=n+1}^{N}\mu\left(\prod_{j:s_{j}=i}\bar{G}_{j}\right)\right\} ,\qquad s\in\left\llbracket 1,N\right\rrbracket ^{n},\:y\in\mathsf{Y}^{n},
\]
with $\bar{G}_{p}(x)=g(x,y_{p})/\nu(y_{p})$ for $p\in\left\llbracket 1,n\right\rrbracket $.
Since $S$ and $Y$ are independent, we consider terms of the form
$\mathsf{E}\left[\psi_{N}(s,Y)\right]$, and define
\[
\psi_{N,q}(s_{q:n},y_{q:n}):=\left\{ \prod_{p=q}^{n}\mu\left(\bar{G}_{p}\prod_{j=q}^{p}\bar{G}_{j}^{\mathbb{I}\{s_{j}=p\}}\right)\right\} \left\{ \prod_{i=n+1}^{N}\mu\left(\prod_{j=q}^{n}\bar{G}_{j}^{\mathbb{I}\{s_{j}=i\}}\right)\right\} ,
\]
for $q\in\left\llbracket 1,n\right\rrbracket $, which satisfies $\psi_{N,1}\equiv\psi_{N}$.
We will show that for $s\in\left\llbracket 1,N\right\rrbracket \times\cdots\times\left\llbracket n,N\right\rrbracket $
and $q\in\left\llbracket 1,n-1\right\rrbracket $,
\begin{equation}
\mathsf{E}\left[\psi_{N,q}(s_{q:n},Y_{q:n})\right]=\mathsf{E}\left[\mu(\bar{G}_{q}^{2})\right]^{\mathbb{I}(s_{q}=q)}\mathsf{E}\left[\psi_{N,q+1}(s_{q+1:n},Y_{q+1:n})\right],\label{eq:exp_decomposition_backwards}
\end{equation}
by considering the cases $s_{q}=q$ and $s_{q}\neq q$. If $s_{q}=q$
then,
\begin{eqnarray*}
\mathsf{E}\left[\psi_{N,q}(s_{q:n},Y_{q:n})\right] & = & \mathsf{E}\left[\mathsf{E}\left[\psi_{N,q}(s_{q:n},Y_{q:n})\mid Y_{q+1:n}\right]\right]\\
 & = & \mathsf{E}\left[\mathsf{E}\left[\mu(\bar{G}_{q}^{2})\psi_{N,q+1}(s_{q+1:n},Y_{q+1:n})\mid Y_{q+1:n}\right]\right]\\
 & = & \mathsf{E}\left[\mu(\bar{G}_{q}^{2})\right]\mathsf{E}\left[\psi_{N,q+1}(s_{q+1:n},Y_{q+1:n})\right],
\end{eqnarray*}
while if $s_{q}\neq q$ then,
\begin{eqnarray*}
 &  & \mathsf{E}\left[\psi_{N,q}(s_{q:n},Y_{q:n})\right]=\mathsf{E}\left[\mathsf{E}\left[\psi_{N,q}(s_{q:n},Y_{q:n})\mid Y_{q+1:n}\right]\right]\\
 & = & \mathsf{E}\left[\mathsf{E}\left[\mu(\bar{G}_{q})\left\{ \prod_{p=q+1}^{n}\mu\left(\bar{G}_{p}\prod_{j=q}^{p}\bar{G}_{j}^{\mathbb{I}\{s_{j}=p\}}\right)\right\} \left\{ \prod_{i=n+1}^{N}\mu\left(\prod_{j=q}^{n}\bar{G}_{j}^{\mathbb{I}\{s_{j}=i\}}\right)\right\} \mid Y_{q+1:n}\right]\right]\\
 & = & \mathsf{E}\left[\psi_{N,q+1}(s_{q+1:n},Y_{q+1:n})\right],
\end{eqnarray*}
where the last equality follows from $\mu(\bar{G}_{q})=1$ and Lemma~\ref{lem:expected_muprod}.
Hence, (\ref{eq:exp_decomposition_backwards}) holds for all $S$
with positive probability and $q\in\left\llbracket 1,n-1\right\rrbracket $.
From $\mathsf{E}\left[\psi_{N,n}(s_{n},Y_{n})\right]=\mathsf{E}\left[\mu(\bar{G}_{n}^{2})\right]^{\mathbb{I}(s_{n}=n)}$
and (\ref{eq:exp_decomposition_backwards}), we conclude that
\[
\mathsf{E}\left[\mathbb{E}\left[\psi_{N}(S,Y)\right]\right]=\prod_{p=1}^{n}\mathbb{E}\left[\mathsf{E}\left[\mu(\bar{G}_{p}^{2})\right]^{\mathbb{\mathbb{I}}(S_{p}=p)}\right]=\prod_{p=1}^{n}\left[1+C/(N-p+1)\right].\qedhere
\]
\end{proof}

\section*{Supplementary materials}

The supplementary materials consist of proofs of Propositions~\ref{prop:simple_analysis},~\ref{prop:biased_analysis},~\ref{prop:identicalGs_momentcond}
and~\ref{prop:indicator_better}.

\section{Proofs of Propositions~\ref{prop:simple_analysis} and~\ref{prop:biased_analysis}}
\begin{proof}[Proof of Proposition~\ref{prop:simple_analysis}]
Let $Z_{p,M}:=\frac{1}{M}\sum_{i=1}^{M}G_{p}(\zeta_{(p-1)M+i})$
for $p\in\left\llbracket 1,n\right\rrbracket $. Since $\mathbb{E}\left[Z_{p,M}\right]=\mu(G_{p})$
for each $p\in\left\llbracket 1,n\right\rrbracket $ and $Z_{1,M},\ldots,Z_{n,M}$
are independent random variables we obtain $\mathbb{E}\left[\gamma_{{\rm simple}}^{N}\right]=\prod_{p=1}^{n}\mu(G_{p})=\gamma$.
As $M\rightarrow\infty$, $Z_{p,M}\overset{P}{\rightarrow}\mu(G_{p})$
for each $p\in\left\llbracket 1,n\right\rrbracket $ by the Weak Law
of Large Numbers, and so $\gamma_{{\rm simple}}^{N}\overset{P}{\rightarrow}\gamma$
as $M\rightarrow\infty$ by Slutsky's Theorem. To obtain the expression
for the second moment \citep[see, e.g.,][]{goodman1962variance} we
have
\[
\mathbb{E}\left[\left(\gamma_{{\rm simple}}^{N}\right)^{2}\right]=\prod_{p=1}^{n}\mathbb{E}\left[Z_{p,M}^{2}\right]=\prod_{p=1}^{n}\left[{\rm var}(Z_{p,M})+\mu(G_{p})^{2}\right],
\]
from which we conclude using ${\rm var}(Z_{p,M})=\left[\mu(G_{p}^{2})-\mu(G_{p})^{2}\right]/M$
and the definition of $\bar{G}_{1},\ldots,\bar{G}_{n}$.
\end{proof}
\begin{proof}[Proof of Proposition~\ref{prop:biased_analysis}]
Let $Z_{p,N}:=\frac{1}{N}\sum_{i=1}^{N}G_{p}(\zeta_{i})$ for $p\in\left\llbracket 1,n\right\rrbracket $.
The Weak Law of Large Numbers provides that $Z_{p,N}\overset{P}{\rightarrow}\mu(G_{p})$
for each $p\in\left\llbracket 1,n\right\rrbracket $ as $N\rightarrow\infty$
and so $\gamma_{{\rm biased}}^{N}\overset{P}{\rightarrow}\gamma$
as $N\rightarrow\infty$ by Slutsky's Theorem. However, we observe
that
\[
\mathbb{E}\left[\gamma_{{\rm biased}}^{N}\right]=\mathbb{E}\left[\prod_{p=1}^{n}\frac{1}{N}\sum_{i=1}^{N}G_{p}(\zeta_{i})\right]=\mathbb{E}\left[\prod_{p=1}^{n}G_{p}(\zeta_{K_{p}})\right],
\]
where $K$ is a vector of $n$ independent ${\rm Uniform}(\left\llbracket 1,N\right\rrbracket )$
random variables and this is not in general equal to $\gamma$.
\end{proof}

\section{Proofs of Propositions~\ref{prop:identicalGs_momentcond} and~\ref{prop:indicator_better}}

To prove Propositions~\ref{prop:identicalGs_momentcond} and~\ref{prop:indicator_better}
we need some additional lemmas.
\begin{lem}[See, e.g., \citealt{esary1967association}]
\label{lem:esary}For any random variable $X$ and non-decreasing
real-valued functions $g_{1}$ and $g_{2}$, $\mathbb{E}\left[g_{1}(X)g_{2}(X)\right]\geq\mathbb{E}\left[g_{1}(X)\right]\mathbb{E}\left[g_{2}(X)\right]$.
\end{lem}
\begin{lem}
\label{lem:indicator-lem-2}Let $B_{1},\ldots,B_{n}$ be Bernoulli
r.v.s with $B_{1}\sim{\rm Bernoulli}(n/N)$ and 
\begin{equation}
B_{p}\mid(B_{1},\ldots,B_{p-1})\sim{\rm Bernoulli}\left(\frac{n-\sum_{j=1}^{p-1}B_{j}}{N-p+1}\right).\label{eq:Bernoulli_assumption}
\end{equation}
Then, for any $m_{1},\ldots,m_{n}$ all greater than or equal to $1$,
$\mathbb{E}\left[\prod_{p=1}^{n}m_{p}^{B_{p}}\right]\leq\prod_{p=1}^{n}\mathbb{E}\left[m_{p}^{B_{p}}\right]$.
\end{lem}
\begin{proof}
For any $\ell\in\left\llbracket 1,n\right\rrbracket $, define $Z_{\ell}=\sum_{p=1}^{\ell}B_{p}$.
From (\ref{eq:Bernoulli_assumption}), $(B_{1},\ldots,B_{\ell-1})$
and $B_{\ell}$ are conditionally independent given $Z_{\ell-1}$
and so
\[
\mathbb{E}\left[-\prod_{p=1}^{\ell}m_{p}^{B_{p}}\right]=\mathbb{E}\left[\mathbb{E}\left[-m_{\ell}^{B_{\ell}}\mid Z_{\ell-1}\right]\mathbb{E}\left[\prod_{p=1}^{\ell-1}m_{p}^{B_{p}}\mid Z_{\ell-1}\right]\right].
\]
We now show that $g_{1}(Z_{\ell-1}):=\mathbb{E}\left[-m_{\ell}^{B_{\ell}}\mid Z_{\ell-1}\right]$
and $g_{2}(Z_{\ell-1}):=\mathbb{E}\left[\prod_{p=1}^{\ell-1}m_{p}^{B_{p}}\mid Z_{\ell-1}\right]$
are non-decreasing functions of $Z_{\ell-1}$ so that $\mathbb{E}\left[g_{1}(Z_{\ell-1})g_{2}(Z_{\ell-1})\right]\geq\mathbb{E}\left[g_{1}(Z_{\ell-1})\right]\mathbb{E}\left[g_{2}(Z_{\ell-1})\right]$
by Lemma~\ref{lem:esary}. That $g_{1}$ is non-decreasing follows
from (\ref{eq:Bernoulli_assumption}) and $m_{p}\geq1$. Interpreting
$(B_{1},\ldots,B_{\ell-1})$ as a draw from a hypergeometric experiment,
we can rewrite $g_{2}(Z_{\ell-1})=\mathbb{E}\left[\prod_{j=1}^{Z_{\ell-1}}m_{I_{j}}\right]$
where $I_{1},\ldots I_{Z_{\ell-1}}$ are drawn uniformly without replacement
from $\left\llbracket 1,\ell-1\right\rrbracket $. Hence, from $m_{p}\geq1$
for all $p\in\left\llbracket 1,\ell-1\right\rrbracket $ and a simple
coupling argument we obtain that $g_{2}$ is also non-decreasing.
It follows that $\mathbb{E}\left[\prod_{p=1}^{\ell}m_{p}^{B_{p}}\right]\leq\mathbb{E}\left[m_{\ell}^{B_{\ell}}\right]\mathbb{E}\left[\prod_{p=1}^{\ell-1}m_{p}^{B_{p}}\right]$
and since $\ell$ is arbitrary we can conclude.
\end{proof}
\begin{proof}[Proof of Proposition~\ref{prop:identicalGs_momentcond}]
Let $m=\mu(\bar{G}_{1}^{2})$ and observe that for every $s$ with
$s_{p}\in\left\llbracket p+1,n\right\rrbracket $
\[
\psi_{N}(s)\leq\left(\prod_{p=1}^{n}m^{\left|\{j:s_{j}=p\}\right|}\right)\prod_{p=n+1}^{N}m^{\left(\left|\{j:s_{j}=p\}\right|-1\right)\vee0}=m^{\sum_{p=1}^{n}b_{p}},
\]
where $b_{p}=\mathbb{I}\left(s_{p}\leq n\text{ or }s_{p}\in\{s_{1},\ldots,s_{p-1}\}\right)$.
Hence, $\mathbb{E}\left[(\gamma_{{\rm recycle}}^{N}/\gamma)^{2}\right]=\mathbb{E}\left[\psi_{N}(S)\right]\leq\mathbb{E}\left[m^{\sum_{p=1}^{n}B_{p}}\right]$
where $S$ is as defined in Theorem~\ref{thm:approx_consistency}.
From Lemma~\ref{lem:indicator-lem-2}, we conclude that
\[
\mathbb{E}\left[(\gamma_{{\rm recycle}}^{N}/\gamma)^{2}\right]\leq\mathbb{E}\left[m^{\sum_{p=1}^{n}B_{p}}\right]\leq\prod_{p=1}^{n}\mathbb{E}\left[m^{B_{p}}\right]=\mathbb{E}\left[(\gamma_{{\rm simple}}^{N}/\gamma)^{2}\right].\qedhere
\]
\end{proof}
\begin{lem}
\label{lem:indicator-lem-1}Let $G_{1},\ldots,G_{n}$ be as in the
statement of Proposition~\ref{prop:indicator_better}. Then for any
$s\in\left\llbracket 1,N\right\rrbracket ^{n}$ such that $s_{p}\geq p$
for all $p\in\left\llbracket 1,n\right\rrbracket $ we have $\psi_{N}(s)\leq\prod_{p=1}^{n}\mu(\bar{G}_{p}^{2})^{b_{p}}$,
where $b_{p}=\mathbb{I}\left(s_{p}\in\left\llbracket 1,n\right\rrbracket \cup\{s_{1},\ldots,s_{p-1}\}\right)$.
\end{lem}
\begin{proof}
Define $m_{p}:=1/\mu(A_{p})$. It follows that $\bar{G}_{p}=m_{p}\mathbb{I}_{A_{p}}$
and $\mu(\bar{G}_{p}^{2})=m_{p}$. Moreover, for $i_{1},\ldots,i_{p}\in\left\llbracket 1,n\right\rrbracket $
we have $\mu(\prod_{j=1}^{p}\bar{G}_{i_{j}})\leq\left(\prod_{j=1}^{p}m_{i_{j}}\right)/\max_{j\in\left\llbracket 1,p\right\rrbracket }m_{i_{j}}$,
with equality if the sets $A_{i_{1}},\ldots,A_{i_{p}}$ are nested,
i.e. $A_{i_{j}}\subseteq A_{i_{k}}$ or $A_{i_{k}}\subseteq A_{i_{j}}$
for distinct $j,k\in\left\llbracket 1,p\right\rrbracket $. Since
$\psi_{N}(s)$ is a non-decreasing function of products of expressions
of the form $\mu(\prod_{j=1}^{p}\bar{G}_{i_{j}})$, we can upper bound
$\psi_{N}(s)$ by assuming henceforth that $A_{1},\ldots,A_{n}$ are
nested, in which case we observe that $\mu(\prod_{j=1}^{p}\bar{G}_{i_{j}})=\left(\prod_{j=1}^{p}m_{i_{j}}\right)/\max_{j\in\left\llbracket 1,p\right\rrbracket }m_{i_{j}}\leq\prod_{j=2}^{p}m_{i_{j}}$.
Plugging this inequality carefully into $\psi_{N}(s)$ using the definition
of $b_{1},\ldots,b_{n}$ gives $\psi_{N}(s)\leq\prod_{p=1}^{n}m_{p}^{b_{p}}$.
\end{proof}
\begin{proof}[Proof of Proposition~\ref{prop:indicator_better}]
From Theorem~\ref{thm:approx_consistency} and Lemma~\ref{lem:indicator-lem-1}
we have 
\[
\mathbb{E}\left[\left(\gamma_{{\rm recycle}}^{N}/\gamma\right)^{2}\right]=\mathbb{E}\left[\psi_{N}(S)\right]\leq\mathbb{E}\left[\prod_{p=1}^{n}\mu(\bar{G}_{p}^{2})^{B_{p}}\right],
\]
where $S$ is as defined in Theorem~\ref{thm:approx_consistency},
and $B_{p}=\mathbb{I}\left(S_{p}\in\left\llbracket 1,n\right\rrbracket \cup\{S_{1},\ldots,S_{p-1}\}\right)$
for $p\in\left\llbracket 1,n\right\rrbracket $. From Lemma~\ref{lem:indicator-lem-2},
\[
\mathbb{E}\left[\prod_{p=1}^{n}\mu(\bar{G}_{p}^{2})^{B_{p}}\right]\leq\prod_{p=1}^{n}\mathbb{E}\left[\mu(\bar{G}_{p}^{2})^{B_{p}}\right]=\prod_{p=1}^{n}\mathbb{E}\left[1+\frac{\mu(\bar{G}_{p}^{2})-1}{N/n}\right]=\mathbb{E}\left[\left(\gamma_{{\rm simple}}^{N}/\gamma\right)^{2}\right].\qedhere
\]
\end{proof}
\bibliographystyle{agsm}
\bibliography{uape}

\end{document}